\documentclass[12pt]{article}
\usepackage[a4paper, margin=2.5cm]{geometry}
\usepackage[utf8]{inputenc}
\usepackage[T2A]{fontenc}
\usepackage{amsmath, amssymb, amsthm, mathrsfs}
\usepackage{mathtools}
\usepackage{bm}
\usepackage{bbm}
\usepackage{mathrsfs}
\usepackage{tikz-cd}
\usepackage{graphicx}
\usepackage{jheppubm}
\usepackage[bbgreekl]{mathbbol}
\usepackage{braket}
\newcommand{\id}{\mathbbm{1}}

\newcommand{\cK}{\mathcal{K}}
\newcommand{\cF}{\mathcal{F}}

\newcommand{\cJ}{\mathcal{J}}
\newcommand{\cL}{\mathcal{L}}
\newcommand{\cO}{\mathcal{O}}
\newcommand{\cR}{\mathcal{R}}
\newcommand{\cS}{\mathcal{S}}
\newcommand{\by}{\bf{y}}
\newcommand{\mbf}{\mathbb {f}}
\newcommand{\mbg}{\mathbb {g}}
\newcommand{\bP}{\mathbb{P}}
\newcommand{\bQ}{\mathbb{Q}}

\newcommand{\bC}{\mathbb{C}}

\newcommand{\bL}{\mathbb{L}}
\newcommand{\bT}{\mathbb{T}}
\newcommand{\bY}{\mathbb{Y}}
\newcommand{\bF}{\mathbb{F}}
\newcommand{\bJ}{\mathbb{J}}
\newcommand{\bR}{\mathbb{R}}
\newcommand{\bU}{\mathbb{U}}

\newcommand{\bS}{\mathbb{S}}

\newcommand{\bPhi}{\mathbb{\Phi}}
\newcommand{\bGamma}{\mathbb{\Gamma}}
\newcommand{\bPsi}{\mathbb{\Psi}}

\newcommand{\bu}{\mathbb{u}}
\newcommand{\bba}{\mathbb{a}}
\newcommand{\bb}{\mathbb{b}}
\newcommand{\bd}{\mathbb{d}}
\newcommand{\bbe}{\mathbb{e}}
\newcommand{\br}{\mathbb{r}}
\newcommand{\bs}{\mathbb{s}}
\newcommand{\bnu}{\mathbb{n}}
\newcommand{\bH}{\mathbb{H}}
\newcommand{\bI}{\mathbb{I}}
\newcommand{\bD}{\mathbb{D}}
\newcommand{\bM}{\mathbb{M}}
\newcommand{\bW}{\mathbb{W}}
\newcommand{\bp}{\mbox{\bf{p}}}
\newcommand{\cP}{\mathcal{P}}
\newcommand{\cD}{\mathcal{D}}

\newcommand{\bG}{\mathbb{G}}

\newcommand{\fL}{\mathfrak{L}}

\newcommand{\Pv}{\mbox{Pv}}

 \newcommand{\sgn}{\mbox{sgn}}
  \newcommand{\mIm}{\mbox{Im}}
   \newcommand{\mRe}{\mbox{Re}}

\newcommand{\tr}{\operatorname{Tr}}

\newcommand{\bk}{\mathbf{k}}
\newcommand{\bx}{\mathbf {x}}

\newtheorem{lemma}{Lemma}
\newtheorem{theorem}{Theorem}
\newtheorem{definition}{Definition}
\theoremstyle{remark}

\definecolor{darkblue}{rgb}{0.,0,1}
\definecolor{dgreen}{rgb}{0,0.6,0}
\definecolor{orange}{rgb}{0.89,0.3,0.12}
\definecolor{oiBlue}{HTML}{0072B2}
\definecolor{oiVermilion}{HTML}{D55E00}
\definecolor{oiBlack}{HTML}{000000}

\newcommand{\IY}{\textcolor{red}}

\newcommand{\be}{\begin{equation}}
\newcommand{\ee}{\end{equation}}
\newcommand{\bea}{\begin{eqnarray}}
\newcommand{\eea}{\end{eqnarray}}
\newcommand{\nn}{\nonumber}

\title{Real Quantum Field Theory, $\bJ$-Quantization, and Standard Model}
\author{I. Aref'eva and I. Volovich}
\affiliation{Steklov Mathematical Institute, Russian Academy of Sciences, \\ Gubkina Street, 8, 119333 Moscow, Russia}
\emailAdd{arefeva@mi-ras.ru}
\emailAdd{volovich@mi-ras.ru}

\abstract{
We present a formulation of quantum field theory based entirely on real
numbers, which we call real quantum field theory (RQFT). The construction is
obtained from the standard complex-number formulation by replacing the imaginary unit $i$ with a matrix $\bJ$ throughout all formulas. Here  $\bJ$ is the real  2$\times$2 matrix  satisfying the condition $\bJ^2=-1$.
This extends the recently developed formulation of real quantum mechanics
based on the real K\"ahler space \cite{Volovich:2025rmi,Arefeva:2025zbx} to systems with infinitely many degrees
of freedom.
\\

As the basic example, we construct the RQFT for
a scalar field. We define a field operator
$
\bPhi(t,\mathbf{x})
$
acting on the real Kähler analogue of the bosonic Fock space. This operator
satisfies the Klein--Gordon equation and the $\bJ$-form of  canonical commutation relation
$
[\bPhi(t,\mathbf{x}),\dot{\bPhi}(t,\mathbf{y})]
=
\bJ\,\delta^{(3)}(\mathbf{x}-\mathbf{y}) .
$
To construct the RQFT we develop the corresponding \(\bJ\)-calculus. In particular, we introduce
the direct and inverse \(\bJ\)-Fourier transforms and the associated
\(\bJ\)-valued distributions. 
In RQFT, the usual unitarity condition for the complex $S$-matrix is replaced
by the statement that the real scattering operator is both orthogonal and
symplectic.
The physical
observables in RQFT coincide with those of ordinary QFT. Thus RQFT
does not change the physical predictions of scalar QFT, but provides an
 equivalent formulation.
\\

We explore the possibility of purely real formulations of the Standard Model of elementary particles. We show that it does admit this formulation and the resulting theory has ortho-symplectic symmetry. 
The choice of  the Standard Model  as the testbed for exploring the possibility of purely real formulations is 
related with the fact it  provides a realistic description of all known fundamental particles. 
The real formulation of the Standard Model naturally suggests a possible exit beyond the Standard Model. In particular, we consider the implications of breaking the $\bJ$-symmetry as a marker for new physics.
}

\begin{document}
\maketitle
\section{Introduction}

The question of whether it is possible to construct a quantum theory using only real numbers has a long history \cite{BVN,Stueckelberg,Varadarajan}. Recently, there has been a revival of interest in this question \cite{Renou:2021dvp,Wu:2022vvi,Chen:2021ril,
Finkelstein:2021zrt,
Chiribella:2022dgr,
Zhu:2020iml,Fuchs:2022rih,
Vedral:2023pij,Sarkar:2025tmd,Hita:2025okv,Hoffreumon:2025nmq,Feng:2025eci,Volovich:2025rmi,Arefeva:2025zbx,Maioli:2026mhy,Ying:2025xyl,Singh:2026dtp,Kalarde:2026fsn,Bang:2026lbt}, motivated primarily by developments in quantum information and computation \cite{Nielsen-book,Ohya-book}.
\\

In this paper, we generalize the results on real k\"ahler quantum mechanics \cite{Volovich:2025rmi,Arefeva:2025zbx} to theories with an infinite number of degrees of freedom.  Real quantum mechanics is constructed using the  quantization rule \cite{Volovich:2025rmi}
\be
\label{JQM}
[\bQ,\bP]=\bJ,\ee
where  $\bJ$ is the 2$\times$2 matrix with real entries satisfying the condition $\bJ^2=-1$. We use natural units $\hbar=c=1$
throughout.  One also has the corresponding Schrodinger equation with $i \to \bJ$. This quantum mechanics is formulated in a real K\"ahler space, instead of the usual formulation of quantum mechanics on a complex Hilbert space.
\\

Generalizing the finite-degree-of-freedom quantum mechanical formalism, which exclusively uses real numbers, to the infinite-dimensional setting of quantum field theory raises a twofold issue. The first aspect deals with quantizing a real scalar field and asks whether quantization can be performed without reference to imaginary unit $i$; in this context, quantization is a generalization of \eqref{JQM} and involves $\bJ$-quantization. The second aspect addresses whether the theory of elementary particles can be modeled, reformulated, or extended solely in terms of real fields.
\\

The primary motivation for formulating quantum field theory over the reals is to enable a unified framework combining elementary particle physics with gravity. Since the standard geometric formulation of gravity is naturally expressed in
terms of real differential-geometric structures, it is useful to ask whether
quantum field theory can also be formulated in a purely real geometric language.
 The Standard Model relies on complex internal symmetry groups, whereas the standard geometric formulation of gravity is naturally expressed in
terms of real differential-geometric structures.
Hence it is useful to ask whether
quantum field theory can also be formulated in a purely real language.
\\

We begin by constructing a quantum field theory for the simplest case of a scalar field, employing only real-valued fields. This involves a generalization of the commutation relation \eqref{JQM}, namely
\be\label{CRJ}
[\mathbb{\Phi}(t,\mathbf{x}),\,{\mathbb\Pi}(t,\mathbf{y})] = \bJ\,\delta(\bx-\by).
\ee
Whereas the usual scalar field is defined as a Hermitian operator in the complex Fock space, in this paper the scalar field
$\mathbb{\Phi}$ and its conjugate $\mathbb\Pi$
are taken to be  operators in a real Kähler Fock space.
\\

 Quantization over the reals is constructed simply by the substitution
\be
i \longmapsto \bJ
\ee
in the usual quantum mechanics and the usual field theories. We refer to the theory based on the $\bPhi$ fields as real quantum field theory (RQFT).
\\

The analogy between the symplectic and complex structures allowed us to develop the so-called $\bJ$-analysis, which amounts to replacing the imaginary unit $i$ with the matrix $\bJ$. In particular, this framework provides forms for the direct and inverse $\bJ$-Fourier transforms, as well as the theory  of $\bJ$-generalized functions and the $\bJ$-Sokhotski--Plemelj formula.   This analysis is related to the theory of analysis over commutative Banach
algebras developed in  \cite{VSV-IV} and related also with p-adic analysis \cite{p-adic}.
\\

We also discuss observables in RQFT and compare them with the corresponding observables in QFT. We have discussed this question in the context of real K\"ahler quantum mechanics \cite{Volovich:2025rmi,Arefeva:2025zbx,Maioli:2026mhy} and found that the corresponding observables coincide with those of usual complex quantum mechanics, see also \cite{Vedral:2023pij,Hita:2025okv,
Hoffreumon:2025nmq}.
In high-energy physics, some observable quantities are related to representations of the Poincar\'e group, while some others are associated with probabilities for scattering processes. In standard QFT, we deal with a unitary Poincar\'e representation, whereas in RQFT we have an orthogonal symplectic Poincar\'e representation. These representations are equivalent and are characterized by the same real numbers.
\\

The most important observables in high-energy physics, however, are not operators in the traditional sense but rather probabilities for scattering processes. These are calculated from scattering amplitudes and are derived from the $S$-matrix. Quantities such as decay rates and cross sections, measured in collider experiments, are derived from these amplitudes.
We introduce the $\bS$-matrix, which, unlike the unitary $S$-matrix of QFT, is an orthogonal and symplectic operator. From this $\bS$-matrix we define decay rates and cross sections and show that they coincide with those derived from the usual $S$-matrix.
\\

The $\bS$-matrix in RQFT is defined by the analog of the Dyson formula
\be
\bS = T \exp\left\{ \bJ \int d^4x \, \mathcal{L}_{\rm int}(\bPhi(x)) \right\}.
\ee
There are also new formulas for the $\bJ$-Feynman propagators.
The usual unitarity condition for the $S$-matrix is replaced by the orthogonality and symplectic conditions for $\bS$, which are realized perturbatively via the $\bJ$-Cutkosky rules.
\\

While RQFT, as an alternative formulation of ordinary QFT, has not  produced new results, it provides an equivalent but conceptually useful
formulation of ordinary quantum field theory. It shows that complex numbers are
not indispensable: their role can be played by a real
matrix complex structure acting on a real Kähler space. This viewpoint may be
helpful both for foundational questions and for the search for formulations of
quantum theory more directly connected with real symplectic geometry. It also offers a unified perspective on classical and quantum theories in terms of symplectic geometry \cite{Volovich-mon,Volovich:2024nzw,Koukoutsis:2026isn}. The interplay between symplectic and complex structures underpins the K\"ahler approach to quantum mechanics developed in \cite{Volovich:2025rmi,Arefeva:2025zbx,Maioli:2026mhy}.
\\

The second part of the paper is concerned with whether the Standard Model the  of elementary particles can be expressed using only real fields. We refer to this approach as realification of the theory.
The Standard Model is chosen for this investigation owing to its exceptional status as a realistic framework for fundamental interactions, which merits careful examination from both phenomenological and conceptual perspectives.
To that end, we begin by considering the realification of its gauge groups—namely U(1), SU(2), and SU(3). While this material is standard in the mathematical literature, we include it in Appendix \ref{app:GSM} for the sake of completeness.
While the Standard Model is usually formulated in terms of Weyl spinors,
 we find it instructive, in addition to the realification of Weyl spinors, to present the realification of Dirac and Majorana spinors;
 this analysis is collected in a dedicated Section \ref{app:fermi} of Appendix. Subsequently, for the realified Weyl spinor, we present the explicit form of the covariant derivative within the Standard Model. We then turn our attention to the realification of the Higgs doublet and provide details of its interactions with the realified spinors. Finally, we briefly outline how the spontaneous symmetry-breaking mechanism operates in this realified version.
\\

As a general remark about realification of a complex  theory by field doubling, the following remark is in order.
Any complex theory can be rewritten over the reals by doubling dimensions. These theories have $\bJ$-symmetry. Reformulation in term of reals provides a possibility to break this symmetry. 
In the Standard Model  the $\bJ$-symmetry is not broken and it is related with $U(1)$-symmetry.  Breaking $\bJ$-symmetry would take the theory beyond the complex-linear Standard Model framework. It is well known that  in the Standard Model  contains no   gauge-invariant bare fermion mass terms before electroweak symmetry breaking. Here we discuss the $\bJ$-symmetry breaking  in context of the neutrino masses.
From the perspective of a fundamentally real-field approach, the number of real fields need not be even, in contrast to complex formulations. Consequently, a complex theory -- which, when reformulated in real terms, is characterized by an even number of real degrees of freedom -- can be embedded into an extended framework with an odd number of real degrees of freedom. 
\\

The paper is organized as follows. In Section~\ref{sec:KF} we extend the real Kähler
formulation of quantum mechanics to the scalar quantum field and construct the
field operator $\bPhi$ acting in the real Kähler Fock space, that is defined in Appendix \ref{sec:KFock}. In
Section~3 we introduce the $\bJ$-Wightman functions and compute the
two-point functions for the massless and massive scalar fields. In Section~\ref{sec:J-Green} we introduce
the $\bJ$-Feynman Green function $\bG_{2,F}(x)$ and compute its $\bJ$-Fourier transform $\tilde\bG^{\bJ}_{2,F}(\bp)$. Section~\ref{sec:obs} discusses observables in RQFT and
compares them with the corresponding observables in standard QFT. In
Section~\ref{sec:JS} we formulate the scattering operator in RQFT and derive the
$\bJ$-analogue of the optical theorem, illustrating it by simple
examples. The appendices collect the mathematical tools used in the main text,
including the real--complex Fock-space correspondence, the $\bJ$-Fourier
transform, the $\bJ$-Sokhotski--Plemelj formula, and the
$\bJ$-Cutkosky rules. Finally in Section \ref{RSM}  we sketch main ingredients of the realification version of  the Standard Model starting from  the fermions of the Standard Model, Section \ref{app:GSM-fermi}, and  including gauge field, Section \ref{RSM-gauge} and the Higgs field, Section \ref{Higgs}. In Appendix \ref{app:GSM} we describe the realification of the groups of the Standard Model and in
 Appendix \ref{app:fermi} the realification of the Dirac, Majorana and Weyl  fermions.  
\\

\newpage

\section{Free Scalar Field in K\"ahler Fock   Space}\label{sec:KF}
The solution of the Klein-Gordon equation 
\be\label{KG}
(\partial^2_t-\Delta _x +m^2)\hat \bPhi=0,\ee
where $\Delta _x$ is the Laplace operator, $m^2\geq 0$, in $\bR ^4$ satisfying the condition \eqref{CRJ} can be written in the form 
\bea\label{phiK}
\bPhi(t,\mathbf{x}) 
= \int \frac{d^3k}{(2\pi)^{3/2}} \frac{1}{\sqrt{2\omega_k}} \left[ a(\mathbf{k}) e^{-\bJ(\omega_k t - \mathbf{k}\cdot\mathbf{x})} + a^\dagger(\mathbf{k}) e^{\bJ(\omega_k t - \mathbf{k}\cdot\mathbf{x})} \right],
\eea
where $\omega_k = \sqrt{\mathbf{k}^2 + m^2}$, and  creation and annihilation operators \( a(\mathbf{k}), a^\dagger(\mathbf{k}')\) satisfy
\begin{equation}\label{aa}
[a(\mathbf{k}), a^\dagger(\mathbf{k}')] = \delta^{(3)}(\mathbf{k} - \mathbf{k}'), \qquad [a(\mathbf{k}), a(\mathbf{k}')]= [a^\dagger(\mathbf{k}), a^\dagger(\mathbf{k}')]=0,
\end{equation}
and
\be\label{bJ}
\bJ = \begin{pmatrix}0&-1\\1&0\end{pmatrix},\qquad \bJ^2 = -\id.
\ee
The form of the matrix $\bJ$ is fixed from the isomorphism between $\bC$ and $\bR^2$ accepted in \cite{Volovich:2025rmi}, see also Lemma 1 in Appendix \ref{sec:C-R2}.

The field  $\hat \bPhi $  acts in $\cK\cF_{\mathbb R}
\left(
L^2_{\mathbb R}(\mathbb R ^3)
\right))$, see definition of  $\cK\cF_{\mathbb R}
\left(
L^2_{\mathbb R}(\mathbb R ^3)
\right))$ in Appendix \ref{sec:KFock}.
\\

We can also write
\be\label{bPhiMM}
\hat\bPhi(f) = \hat\bPhi^-(f)+
\hat\bPhi^+(f),
\ee
where
\bea\label{bF+}
\hat\bPhi^-(f)=\int  \frac{d^3k}{(2\pi)^{3/2}\sqrt{2\omega_k}} \,a(\bk)\,\tilde\mbf(\bk),\quad \hat\bPhi^+(f)=\int  \frac{d^3k}{(2\pi)^{3/2}\sqrt{2\omega_k}}\,a^\dagger(\bk)\,\tilde\mbf^T(k)\,
\eea
and 
\bea\label{mbf}
\tilde\mbf(\bk)&=& \int d^4x f(x)e^{-\bJ(\omega_k t - \mathbf{k}\cdot\mathbf{x})}, \quad 
\tilde\mbf^T(\bk)=  \int d^4x f(x)e^{\,\bJ(\omega_k t - \mathbf{k}\cdot\mathbf{x})}.
\eea

 $\hat\bPhi(f)$ is a 2$\times$2 matrix with entries being the operators. Note, that if $f(x)\in \cS(\bR^4)$, the formula \eqref{mbf} defines  $ \mbf (\bk)\in \cS(\bR^3)$, see Appendix \ref{app:JFT2}, theorem \ref{th:3}.
 \\
 
The time derivative of the scalar field \eqref{phiK} is

\bea\label{phi-der}
&&\partial_t\hat\bPhi(t,\mathbf{x}) = \int \frac{d^3k}{(2\pi)^{3/2}} \frac{\bJ \omega_k}{\sqrt{2\omega_k}} \left[- a(\mathbf{k}) e^{-\bJ(\omega_k t - \mathbf{k}\cdot\mathbf{x})} + a^\dagger(\mathbf{k}) e^{\bJ(\omega_k t - \mathbf{k}\cdot\mathbf{x})} \right]
\eea
and $\partial_t\hat\bPhi(f)$ defined as 
\bea
&&\partial_t\hat\bPhi(f) =
\int d^4x f(x)\partial_t\hat\bPhi(t,\mathbf{x})
\eea
can be represented as
\bea
\partial_t\hat\bPhi(f)\nn&=&
\bJ\int\frac{d^3k}{(2\pi)^{3/2}\sqrt{2\omega_k}} \, \,\omega_k[-a(\bk) \mbf(\bk) +a^\dagger(\bk)\mbf^T(\bk)]\eea
Therefore,
\bea\label{derPhi}
[\hat\bPhi(f),\partial_t\hat\bPhi(g)]
=\frac{\bJ}{2}\int \frac{d^3k}{(2 \pi)^3} \,\Big(\mbf(\bk)\,
\mbg^T(\bk)+\mbf^T(\bk)\mbg(\bk)\Big)\eea
\\

We can get the field  $\hat \bPhi$ by mapping  the usual scalar quantum field $\hat \phi $ acting in $\cF_{\mathbb C}
\left(
L^2_{\mathbb C}(\mathbb R ^3)
\right))$ to the field acting in $\cK\cF_{\mathbb R}
\left(
L^2_{\mathbb R}(\mathbb R ^3)
\right))$, i.e. 
\be\label{gamma-phi}
\cR\Big(\hat \phi (t,\bx)\Big)=\hat \bPhi (t,\bx),\ee
where
\begin{equation}\label{phi}
\hat\phi(t,\mathbf{x}) = \int \frac{d^3k}{(2\pi)^{3/2}} \frac{1}{\sqrt{2\omega_k}} \left[ a(\mathbf{k}) e^{-i(\omega_k t - \mathbf{k}\cdot\mathbf{x})} + a^\dagger(\mathbf{k}) e^{i(\omega_k t - \mathbf{k}\cdot\mathbf{x})} \right].
\end{equation}

In the sense of operator valued distributions
\be
\cR\Big(\hat\phi(f)\Big)=\hat\bPhi(\bF),\ee
where
\bea\label{phi-f}
\hat\phi(f)&=&\int d^4x f(x)\hat\phi(t,\mathbf{x}), \quad 
\hat\phi(f)=\hat\phi^-(f)+\hat\phi^+(f), 
\eea
with
\begin{align} \hat\phi^-(f)&=\int \frac{d^3k}{(2\pi)^{3/2}} \frac{1}{\sqrt{2\omega_k}} a(\bk)\tilde f (\bk),\quad 
\hat\phi^+(f)=\int \frac{d^3k}{(2\pi)^{3/2}} \frac{1}{\sqrt{2\omega_k}} a\dagger(\bk)\tilde f ^*(\bk),\\\label{ffs}
\tilde f (\bk)&=\int d^4x f(x)e^{-i(\omega_k t - \mathbf{k}\cdot\mathbf{x})},\quad \qquad 
\tilde f ^*(\bk)=\int d^4x f(x)e^{\IY{+}i(\omega_k t - \mathbf{k}\cdot\mathbf{x})}.
\end{align}
Note, that if $f(x)\in \cS(\bR^4)$, the formula \eqref{ffs} defines  $\tilde f (\bk),\tilde f ^*(\bk)\in \cS(\bR^3)$, see Appendix \ref{app:JFT2}, theorem \ref{th:2}.
\\

 Also note, that the commutation relation \eqref{derPhi} is the generalization of commutation relation 
\bea\label{scr}
[\hat\phi(f),\partial_t\hat\phi(g)]=\frac{i}{2}\int \frac{d^3k}{(2 \pi)^3} \,\Big(\tilde f(\bk)\,
\tilde g^*(\bk)+\tilde f^*(\bk)\tilde g(\bk)\Big)
\eea


\newpage
\section{ $\bJ$-Wightman functions}\label{sec:wightman}

\subsection{Two-point $\bJ$-Wightman Functions }
Taking into account \eqref{bPhiMM} we get
\bea\label{bFG}
\bW_{2}(f,g)=\langle\!\langle0|\bPhi(f)\bPhi(g)|0\rangle\!\rangle
=\int \frac{d^3k}{(2\pi)^{3}2\omega_k}\,\mbf(\bk) \mbg^T(\bk),\eea
$\mbf(\bk)$ and $\mbg^T(\bk)$ are defined in \eqref{mbf}. 
Here $|0\rangle\!\rangle$ is annihilated by $a(\bk)$,
\be\label{a0}
a(\bk)|0\rangle\!\rangle=0\ee
and
\be\label{alpha}
\langle\!\langle 0| (\alpha \id +\beta\bJ)|0\rangle\!\rangle=\alpha\id +\beta\bJ,\quad \alpha, \beta \in \bR.\ee

Note that the expression in \eqref{bFG} is the $2\times 2 $ matrix which has the form
\be
\bW_{2}(f,g)=\bW_{2,\id}(f,g)
+\bJ \,\bW_{2,\bJ}(f,g)\ee
and
\bea
\bW_{2,\id}(f,g)&=& \int \frac{d^3k}{(2\pi)^{3}2\omega_k}\,\frac {1}{2} \tr[\mbf(\bk) \mbg^T(\bk)]\label{bojborr}\\
\bW_{2,\bJ}(f,g)&=& -\int \frac{d^3k}{(2\pi)^{3}2\omega_k}\,\frac {1}{2} \tr[\mbf(\bk) \mbg^T(\bk)\bJ]\label{boJborr}\eea

$\gamma$ maps the $(\bW_{2}(f,g)$ to the usual two-point Wightman function
\be
\gamma (\bW_{2}(f,g))=W_2(f,g),\ee
where $W_2(f,g)$ is the usual 
 two-point Wightman function
 \bea\label{phi(f)phi(g)}
W_2(f,g)&=&\langle0|\hat\phi(f)\hat\phi(g)|0\rangle=\langle0|\hat\phi^-(f)\hat\phi^+(g)|0\rangle
\\\nn
&=&\int \frac{d^3k}{(2\pi)^{3}2\omega_k}\,\tilde f(\bk) \tilde g^*(\bk),\eea
where $\tilde f (\bk)$ and $\tilde g ^*(\bk)$ are defined by \eqref{ffs}
and it is assumed that $f(x),g(x)$ are real functions. In the usual case 
\be
\label{2W}
W_2(f,g)=\int d^4 x d^4y f(x)g(y)W_2(x,y),\ee
where
\bea\label{W2}
W_2(x;y) =\langle0|\hat\phi(x_0,\mathbf{x})\hat\phi(y_0,\mathbf{y})|0\rangle=\int \frac{d^3p}{(2\pi)^3 \, 2\omega_p} \, e^{-ip\cdot(x-y)}, 
\eea
here $p\cdot x = \omega_p x^0 - \mathbf{p}\cdot \mathbf{x}$.
\\

The $\bJ$-analogue of formula \eqref{2W} is 
\be
\label{2W}
\bW_2(f,g)=\int d^4 x d^4y f(x)g(y)\bW_2(x,y),\ee
where
\bea\label{WW2}
\bW_2(x;y) =\int \frac{d^3p}{(2\pi)^3 \, 2\omega_p} \, e^{-\bJ p\cdot(x-y)}. 
\eea
The regularized version of this formula is
\bea\label{WW2epsilonm}
\bW_{2,\epsilon}(x;y;m) =\int \frac{d^3p}{(2\pi)^3 \, 2\omega_p} \,e^{-\epsilon |p|}\, e^{-\bJ p\cdot(x-y)}. 
\eea
\subsubsection{Two-point $\bJ$-Wightman function for massless field}

Let us compute the $\bJ$-Wightman function for a massless scalar field in \(3+1\) dimensions with the regularization:
\bea\label{bW2epsilon}
 \bW_{2,\epsilon}(x^0,\vec x;y^0,\vec y)&=&\int\frac{d^3\vec{p}}{(2\pi)^3 \, 2|\vec{p}|}e^{-\epsilon |p|}e^{-\bJ p\cdot (x-y)},
\eea
here $p\cdot (x -y)= |\bp| (x^0 -y^0) -\mathbf{p}\cdot (\bx-\by)$.
For simplicity consider the case $y=0$ and we denote $ \bW_{2,\epsilon}(x^0,\vec x;0,0)\equiv \bW_{2,\epsilon}(x^0,\vec x)$.
Using spherical coordinates 
$d^3\vec{p}=p^2 dp\,d\Omega$, $p=|\vec{p}|$, $r=|\vec{x}|$, and $\vec{p}\cdot\vec{x}=pr\cos\theta$ after angular integration,
\bea\label{anglesJ}
 \int d\Omega\, e^{\bJ pr\cos\theta} 
= 4\pi\,\frac{\sin(pr)}{pr}
\eea
we get
 \bea
  \bW_{2,\epsilon}(x^0,r)
  &=& \frac{1}{4\pi^2 r} \int_0^\infty dp\; e^{-\epsilon|p|}(\cos(p x^0)-\bJ \sin(p x^0) ]\sin(pr)
\eea
Performing integration over momentum  we finally get
\bea\label{3.54}
\bW_{2,\epsilon}(x^0,r)&=&\frac{1}{4\pi^2 }\frac{\epsilon ^2+r^2-x_0^2-2 \epsilon  \bJ
   x_0}{\left(\epsilon
   ^2+\left(r-x_0\right){}^2\right) \left(\epsilon
   ^2+\left(r+x_0\right){}^2\right)}\nn\\
 &\approx& \frac{1}{4\pi^2 }\frac{r^2-x_0^2-2 \epsilon  \bJ
   x_0}{(r^2-x_0^2)^2+\epsilon
   ^2(r^2+x_0^2)}\eea
   and taking into account \eqref{4.45} we get
   \bea\label{3.55}
\lim_{\epsilon \to 0}\bW_{2,\epsilon}(x^0,\bx)&=&-\frac{1}{4\pi^2 }\cP\Big(\frac{1}{x_0^2-\bx^2} \Big)-\frac{1}{4\pi}
\bJ\,\sgn  x^0
\delta(x^2)
\eea
  \bea
  \left(\left(\epsilon +\bJ x_0\right)^2+r^2 \right)\bW_{2,\epsilon}(x^0,r)=\frac{1}{4\pi^2 },\eea
  that is equivalent to
  \be
  \bW_{2,\epsilon}(x^0,r)=\frac{1}{4\pi^2 }\, \frac {1}{r^2+\left(\epsilon +\bJ x_0\right)^2},\ee

  Hence 
  \bea\label{3.59}
 \bW_{2}(x;y)&=&\lim_{\epsilon\to 0}  \bW_{2,\epsilon}(x;y)=\frac{1}{4\pi^2 }\, \frac {1}{{(\bx-\by)^2} -(x_0-y_0)^2+\bJ \,\sgn (x_0-y_0)\, 0}\\ \eea
Comparing with the usual Wightman function $W_2$
 \bea\label{3.60}
 W_{2}(x;y)=\lim_{\epsilon\to 0}  W_{2,\epsilon}(x;y)=
 \frac{1}{4\pi^2 }\, \frac {1}{{(\bx-\by)^2} -(x_0-y_0)^2+i \,\sgn (x_0-y_0)\, 0}\eea
   we see that formula \eqref{3.59} is obtained from \eqref{3.60} simply by replacing $i$ with $\bJ$.
\subsubsection{Two-point $\bJ$-Wightman function for massive field}
The 2-point $\bJ$-Wightman function for massive field with  the regularization \eqref{WW2epsilonm} can be represented as
\be
\bW_{2,\epsilon}(x;y;m)
=
\frac{m}{4\pi^2}
\frac{
K_1\!\left(m\sqrt{-(x-y)^2+2\bJ\epsilon\,(x^0-y^0)+\epsilon^2}\right)
}{
\sqrt{-(x-y)^2+2\bJ\epsilon\,(x^0-y^0)+\epsilon^2}
},
\ee
and the  massive Wightman function can be written as a distributional boundary value
\be
\bW_2(x;y;m)
=
\frac{m}{4\pi^2}
\frac{
K_1\!\left(m\sqrt{-X^2+\bJ\,0\,(x^0-y^0)}\right)
}{
\sqrt{-(x-y)^2+\bJ\,0\,(x^0-y^0)}
}.
\ee

For spacelike separation, we have
\be\label{Wms}
\lim _{\epsilon \to 0}\bW_{2,\epsilon}(x;y;m)
=
\frac{m}{4\pi^2}
\frac{
K_1\!\left(m\sqrt{-(x-y)^2}\right)
}{
\sqrt{-(x-y)^2}
}\,\mathbf ,
\qquad (x-y)^2<0,
\ee
and for timelike separation,
\bea\label{Wmt}
\lim _{\epsilon \to 0}\bW_{2,\epsilon}(x;y;m)
&=&
\frac{m}{8\pi\sqrt{(x-y)^2}}
Y_1\!\left(m\sqrt{(x-y)^2}\right)
\\\nn&+&
\frac{m}{8\pi\sqrt{(x-y)^2}}
\operatorname{sgn}(x^0-y^0)\bJ\,
J_1\!\left(m\sqrt{(x-y)^2}\right),
\qquad (x-y)^2>0 .
\eea
here $J_1$
 and $Y_1$
 are the two standard  Bessel functions of order 1. Both $J_1(z)$
 and $Y_1(z)$ are real for positive real argument $z>0$ \cite{Abramowitz:1964}.

Near the light cone, $(x-y)^2\to 0$, the massive Wightman function has the
distributional expansion
\bea\label{Wmlc}
\lim _{\epsilon \to 0}\bW_{2,\epsilon}(x;y;m)
&=&
\frac{1}{4\pi^2}
\frac{1}{-(x-y)^2+\bJ\,0\,(x^0-y^0)}
\\
&
+&
\frac{m^2}{8\pi^2}
\left[
\log\left(
\frac{m}{2}
\sqrt{-(x-y)^2+\bJ\,0\,(x^0-y^0}
\right)
+\gamma_E-\frac12
\right]
\\\nn&+&
\cO\!\left((x-y)^2\log (x-y)^2\right)
\qquad \mbox{for}\quad (x-y)^2\to 0 .
\eea
Thus the leading light-cone singularity of the massive Wightman function
coincides with the massless Wightman distribution, but the full massive
Wightman function contains additional mass-dependent logarithmic terms.


The answers \eqref{Wms}, \eqref{Wmt} and \eqref{Wmlc} are generalization of the usual QFT answers that can be obtained by the replacement $\bJ \longmapsto i$.
\subsection{n-point $\bJ$-Wightman Functions }
n-point $\bJ$-Wightman functions are defined as  
\bea
&&\bW_{n}(f_1,\ldots f_{n})=\langle\!\langle0|\bPhi(f_1)\ldots\bPhi(f_{n})|0\rangle\!\rangle
\label{2nW}
,\eea
As in the usual case n-point functions factorize into the sum  of the product of 2-point functions
:
\be
\bW_{n}(f_1,\ldots f_{n})=
\begin{cases}
0, & n \text{ odd}, \\[6pt]
\displaystyle \sum_{\text{pairings } P} \; \prod_{(i,j)\in P} \bW_2(f_i, f_j), & n = 2m \text{ even},
\end{cases}
\ee
where $\bW_2(f,g) $ are given by \eqref{bFG} and the sum runs over all distinct ways to partition the set $\{1,2,\dots,2m\}$ into $m$ unordered pairs.

\newpage
\section{Feynman $\bJ$-Green Functions}\label{sec:J-Green}

\subsection{Two-point Feynman $\bJ$-Green Functions}
The  2-point Feynman $\bJ$-Green function is defined as
\bea\label{bG}
&&\qquad\qquad\qquad\bG(1,2)\equiv \bG(t_1,\mathbf{x}_1;t_2,\mathbf{x}_2)\\\nn&=&\theta(t_1-t_2)\langle\!\langle0|{\mathbb\Phi}(t_1,\mathbf{x}_1){\mathbb\Phi}(t_2,\mathbf{x}_2)|0\rangle\!\rangle+\theta(t_2-t_1)\langle\!\langle0|{\mathbb\Phi}(t_2,\mathbf{x}_2){\mathbb\Phi}(t_1,\mathbf{x}_1)|0\rangle\!\rangle,
\eea
where $\theta(t)$ is the Heaviside function, $\theta(t)=1$ for $t>0$, $\theta(t)=0$ for $t<0$.
Using \eqref{phiK},\eqref{a0} and \eqref{alpha}
we get
 \be\label{bGF}
\bG_F(x_1;x_2) = \int \frac{d^3k}{(2\pi)^3}\frac{1}{2\omega_k}\left[ \theta(T) e^{-\bJ\omega_k T + \bJ(\mathbf{k}\cdot\mathbf{R})} + \theta(-T) e^{\bJ\omega_k T - \bJ(\mathbf{k}\cdot\mathbf{R})} \right],
\ee
where
\be  T = t_1-t_2, \quad \mathbf{R} = \mathbf{x}_1-\mathbf{x}_2, \quad \omega_k = \sqrt{\bk^2+m^2},
\ee
To find the   $\bJ$-Fourier transform of $\bG_F(x_1;x_2)$ we make a regularization 
\be\label{FbGr}
\tilde{\bG}^\bJ_{ F,\epsilon}(p) = \int d^3 \bx\,\Big[\int_0^\infty dT  e^{-\bJ p\cdot x-\epsilon T}+ \int_{-\infty}^0 dT  e^{-\bJ p\cdot x+\epsilon T}\Big]\,\bG_F(x),
\ee
Substituting $\bG_F(x)$ given by \eqref{bGF} and interchanging integrals gives
\bea\nn
\tilde{\bG}^\bJ_{F,\epsilon}(p) &=& \int \frac{d^3k}{(2\pi)^3}\frac{1}{2\omega_k}\Big[ 
\int_0^\infty dT\, e^{-\bJ (p^0+\omega_k)T-\epsilon T} \int d^3R\, e^{\bJ (\mathbf{p}+\mathbf{k})\cdot\mathbf{R}}\\ &+& \int_{-\infty}^0 dT\, e^{-\bJ (p^0-\omega_k)T+\epsilon T} \int d^3R\, e^{\bJ (\mathbf{p}-\mathbf{k})\cdot\mathbf{R}} \Big].
\eea

The spatial integrals, see  \eqref{D4} in Appendix, yield
\be
\int d^3R\, e^{\bJ \mathbf{q}\cdot\mathbf{R}} = (2\pi)^3 \delta^{(3)}(\mathbf{q})\,\id,
\ee
so we obtain
\be
\tilde{\bG}^\bJ_{F,\epsilon}(p) = \frac{1}{2\omega_p}\left[ \int_0^\infty dT\, e^{-\bJ (p^0+\omega_p)T-\epsilon T} + \int_{-\infty}^0 dT\, e^{-\bJ (p^0-\omega_p)T+\epsilon T} \right].
\ee
Changing the variable $u=-T$ in the second integral gives
\be\label{4.29}
\tilde{\bG}^\bJ_{ F,\epsilon}(p) = \frac{1}{2\omega_p} \int_0^\infty dT\left[\, e^{-\bJ (p^0+\omega_p)T-\epsilon T} +  e^{\bJ (p^0-\omega_p)T-\epsilon T} \right].
\ee

We represent the integrands in \eqref{4.29} as 
\bea\label{4.30}
\tilde{\bG}^\bJ_{F,\epsilon}(p) &=& \frac{1}{2\omega_p} \int_0^\infty dT
\,\Big[ \Big(\cos(p^0+\omega_p)T+\cos((p^0-\omega_p)T\Big)e^{-\epsilon T}  \\\nn
&&\qquad \quad\quad -\bJ\Big(\sin(p^0+\omega_p)T -\sin(p^0-\omega_p)T \Big)e^{-\epsilon T} \Big]
\eea and performing the integrations we get 
\bea\label{JGint}
\tilde{\bG}^\bJ_{F,\epsilon}(p)&=& \label{4.31}\frac{1}{2\omega_p} 
\,\Big[\epsilon  \left(\frac{1}{\epsilon ^2+\left(p_0+\omega
   \right){}^2}+\frac{1}{\epsilon ^2+\left(p_0-\omega
   \right){}^2}\right)\\\nn
&\quad&\qquad+\bJ\Big(-\frac{p_0+\omega }{\epsilon
   ^2+(p_0+\omega )^2}+\frac{p_0-\omega }{\epsilon
   ^2+(p_0-\omega )^2}\Big)\Big]\eea
Using that in the sense of distribution one has identities \eqref{4.45} (see \cite{VVS})
we get for the first line in \eqref{4.31}
\bea
&&\lim_{\epsilon\to 0}  \frac{1}{2\omega_p}\left(\frac{\epsilon}{\epsilon ^2+\left(p_0+\omega
   \right){}^2}+\frac{\epsilon}{\epsilon ^2+\left(p_0-\omega
   \right){}^2}\right)\\
   &=& \frac{\pi}{2\omega_p}\Big(\delta(p_0+\omega)+\delta(p_0-\omega)\Big)=\pi\delta(\omega_p^2-p^2_0)
   \eea
and for
the second line one in \eqref{4.31}
\bea
&&\lim_{\epsilon\to 0}  \frac{1}{2\omega_p}\left(\bJ\Big(-\frac{p_0+\omega }{\epsilon
   ^2+(p_0+\omega )^2}+\frac{p_0-\omega }{\epsilon
   ^2+(p_0-\omega_p )^2}\Big)\right)\\
 && =\bJ \frac{1}{2\omega_p}\left(-\cP(\frac{1}{p_0+\omega_p})+\cP(\frac{1}{p_0-\omega_p})\right)=-\bJ \,\cP(\frac{1}{\omega^2_p -p^2_0}).\eea
Hence finally we get
\bea\label{4.30}
\lim _{\epsilon \to 0}\tilde{\bG}^\bJ_{ F,\epsilon}(p)=\tilde{\bG}^\bJ_{F}(p)
,
\eea
where 
\bea \label{bGG} \tilde{\bG}^\bJ_{F}(p)
=-\bJ \,\cP(\frac{1}{\omega^2_p -p^2_0})&+&\pi\delta(\omega_p^2-p^2_0)=\bJ \,\cP(\frac{1}{p^2-m^2})+\pi\delta(m^2-p^2)
\eea 

Note, that we can get the same result \eqref{bGG} stating from
\bea
\label{bGeps} \tilde{\bG}^\bJ_{ F,simple,\varepsilon}(p) =\frac{\mathbb{J}}{p^2 - m^2 + \bJ\varepsilon^2},
\eea
i.e.
\bea
\label{bGr} \tilde{\bG}^\bJ_{F}(p)= \lim _{\varepsilon \to 0}\tilde{\bG}^\bJ_{F,simple,\varepsilon}(p)
\eea
Indeed, multiplying the numerator and denominator of \eqref{bGeps}  on $p^2 - m^2 - \bJ\varepsilon^2$
we get
\bea
\label{bGeps} \tilde{\bG}^\bJ_{ F,simple,\varepsilon}(p) =\frac{\mathbb{J}(p^2 - m^2 - \bJ\epsilon^2)}{(p^2 - m^2)^2 + \epsilon^4}=\frac{\varepsilon^2+\bJ(p^2 - m^2 )}{(p^2 - m^2)^2 + \varepsilon^4},
\eea
and according \eqref{4.45} we finally get
\bea
\label{bGepss} 
\lim_{\epsilon \to 0}\tilde{\bG}^\bJ_{F,simple,\varepsilon}(p) &=&\lim_{\varepsilon \to 0}\frac{\varepsilon^2}{(p^2 - m^2)^2 + \varepsilon^4}+\bJ\lim_{\varepsilon \to 0}\frac{(p^2 - m^2 )}{(p^2 - m^2)^2 + \varepsilon^4}\\\nn\\&=&\pi \delta(p^2 - m^2)+ \bJ\cP\Big(\frac{1}{p^2 - m^2 }\Big),
\eea
i.e. we get \eqref{bGG}.

It is instructive to compare the propagators for fixed regularizations $\epsilon$ and $\varepsilon$.
For the first line of 
\eqref{JGint} we have
\bea\label{ere}
\frac{\epsilon}{2\omega_p}   \left(\frac{1}{\epsilon ^2+\left(p_0+\omega
   \right){}^2}+\frac{1}{\epsilon ^2+\left(p_0-\omega
   \right){}^2}\right)
 & \approx&
   \frac{2 \epsilon  \left(p_0^2+\omega ^2\right)}{\left(p^2-m^2\right)^2+2 \epsilon ^2
   \left(p_0^2+\omega ^2\right)}
   \eea
   and for the second  line of 
\eqref{JGint} we have
   \bea\label{ire}
   \bJ\frac{1}{2\omega_p} \Big(-\frac{p_0+\omega }{\epsilon
   ^2+(p_0+\omega )^2}+\frac{p_0-\omega}{\epsilon
    ^2+(p_0-\omega )^2}\Big)
  \approx \bJ \frac{ \left(p^2-m^2\right)}{\left(p^2-m^2\right)^2+2 \epsilon ^2
   \left(p_0^2+\omega ^2\right)}
   \eea
Comparing \eqref{ere} and \eqref{ire} with
\bea
\label{bGs} \tilde{\bG}^\bJ_{ F,simple,\varepsilon}(p) =\frac{\varepsilon}{(p^2 - m^2)^2+\varepsilon ^2}+\bJ\frac{(p^2 - m^2)}{(p^2 - m^2)^2+\varepsilon ^2},
\eea
we see that the second term in RHS of \eqref{bGs} coincides with \eqref{ire}
after identification 
\be
\label{ident}2 \varepsilon ^2=\epsilon ^2
   2(p_0^2+\omega ^2).
   \ee
  The first term in 
in RHS of \eqref{bGs} after identification
\eqref{ident}  is different from the last term in \eqref{ere}
 by a factor that is equal to 1 on mass shell $m^2=p^2$.

\subsection{$\cR$ Map of the Standard Complex Answer}
By applying $\cR$ to the standard complex propagator:
\be
\tilde G_{F}(k)=\frac{i}{k^2-m^2+  i0}=i\cP( \frac{1}{k^2-m^2})+\pi \delta(k^2-m^2),\nn\\\label{GF}\ee
we get
\bea\nn
\cR\!\left(\tilde G_F(k)\right)&=& \cR\!\left( \frac{i}{p^2 - m^2 + i0} \right) = \frac{\mathbb{J}}{p^2 - m^2 + \bJ 0}\\
 &=& 
 \bJ\,\cP(\frac{1}{p^2 - m^2 })+\pi \delta(p^2 - m^2 )
 \label{bG}
\eea

Therefore, we have proved that 
 \be
 \tilde\bG^{\bJ}_{F}(p) =\cR\!\left(\tilde G_{F}(p)\right),\ee


\newpage
\section{Observables in RQFT versus Observables in QFT}\label{sec:obs}

In QFT, the quantities that are typically observed are not the field operators themselves at a point, but rather matrix elements, probabilities, spectra, and expectation values of observables. The most important experimentally observed quantities in QFT are the following.

\begin{itemize}
\item Masses and particle quantum numbers.
Particles are classified by eigenvalues of the Poincaré generators,
$
P_\mu P^\mu = m^2,
$ 
and by spin or helicity. Thus one observes
$
m,\ s,\ q,\ B,\ L,\ \dots,
$
where \(q\) is the electric charge, \(B\) the baryon number, \(L\) the lepton number, etc. These quantities are the same in both QFT and RQFT.

\item Expectation values of local observables.
Given a state \(|\Psi\rangle\), one considers
$
\langle \Psi | \mathcal{O}(x) | \Psi \rangle,
$
where \(\mathcal{O}(x)\) is a local observable.
In gauge theories, physical observables must be gauge-invariant; examples include
$
\langle \Psi |F^{\mu\nu}F_{\mu\nu}|\Psi\rangle$,
$
\langle \Psi |\bar{\psi}\psi |\Psi\rangle$, $\mathrm{Tr}\, W(C)
$, here $W(C)$
is the Wilson loop
\begin{equation}
W(C) = \mathcal{P} \exp\left(i \oint_C A_\mu \, dx^\mu \right),
\end{equation}
where \(\mathcal{P}\) denotes path ordering. These quantities are real numbers, and there is no difference between their values in QFT and RQFT.

\item 
The correlation functions
$
G_n(x_1,\dots,x_n) = \langle 0 | T\{\phi(x_1)\cdots\phi(x_n)\} | 0 \rangle
$
are not directly measured; however, via the LSZ reduction formula they determine the \(S\)-matrix, whose matrix elements
$
\langle f | S | i \rangle
$
are related to measurable quantities (see below).

\item Scattering amplitudes.
The \(S\)-matrix provides the transition amplitudes
$
S_{fi} = \langle f | S | i \rangle$.
The amplitude itself is not directly a probability, but it determines measurable probabilities.
\begin{itemize}
\item Cross sections.
For a scattering process
$
a + b \to 1 + \cdots + n$,
the observable quantity is the differential cross section
\be\label{dsigma}
d\sigma \sim |\mathcal{M}_{fi}|^2 \, d\Pi_f,
\ee
where \(\mathcal{M}_{fi}\) is the invariant amplitude (generally a complex number in QFT), and \(d\Pi_f\) is the positive final-state phase-space measure. In RQFT this formula is modified, but it yields the same result for \(d\sigma\) (see below).

\item Decay widths and lifetimes.
For a decay
$
A \to f$,
one measures the partial decay width
$
\Gamma(A \to f)$,
and the total width
\be
\Gamma_A = \sum_f \Gamma(A \to f).
\ee
The lifetime is then
$
\tau_A = \frac{1}{\Gamma_A}$.
\end{itemize}

\item Number of particles and spectra.
In scattering experiments, one observes stable outgoing particles and deals with asymptotic number operators
$
N_p = a_p^\dagger a_p$.
Measurable distributions include
$
\frac{dN}{d^3p}$, $\frac{dN}{dy}$, $ \frac{dN}{dp_T}$.
For unstable particles, the unstable particle itself is not observed as an exact asymptotic state; instead, one observes its decay products and reconstructs the resonance.
\end{itemize}
\newpage
\newpage
\section{$\bS$-matrix}\label{sec:JS}
The $\bS[g]$ matrix in the RQFT  is defined by the analog of the Dyson formula
\be \bS=T\exp \{\,\bJ \int d^4x g(x)L_{int}(\bPhi(x))\}\ee
with  \(g(x)\) is a smooth switching function describing where and when the
interaction is turned on. The physical \(S\)-matrix is recovered formally in
the adiabatic limit
$
g(x)\to 1$. For example, in $\lambda 
\phi^4$
 theory, $L_{int}(\bPhi(x))=\bPhi(x)$, where $\bPhi(x)$ is the real quantum scalar field defined in $\cK\cF_{\mathbb R}
\left(
L^2_{\mathbb R}(\mathbb R )
\right))$ by \eqref{phiK}.

As in the Bogoliubov--Shirkov formulation \cite{BS} we treat the  
$
\bS[g],
$
as a functional on the space of test functions with compact support, i.e. $g\in \ \cD({\mathbb R}^4)$. This is the starting point of causal perturbation theory, where the interaction is first switched off at infinity and only later the adiabatic limit is considered.

The perturbative expansion is written as
\be
\bS[g]
=
\mathbf{1}
+
\sum_{n=1}^{\infty}
\frac{1}{n!}
\int d^4x_1\cdots d^4x_n\,
\bS_n(x_1,\ldots,x_n)
g(x_1)\cdots g(x_n).
\ee
We also supposed that the following axioms are fulfilled.
\begin{itemize}
\item  Normalization.
If there is no interaction, then there is no scattering:
$
S[0]=\mathbf{1}$.
\item Relativistic covariance.
Under a Poincaré transformation
\[
x\mapsto \Lambda x+a,
\]
one requires
\[
\bU(\Lambda,a)\,\bS[g]\,\bU^{-1}(\Lambda,a)
=
\bS[g_{\Lambda,a}],
\]
where
$
g_{\Lambda,a}(x)
=
g\!\left(\Lambda^{-1}(x-a)\right).
$
In the formula above,
$
\bU(\Lambda,a)
$
is the orthogonal symplectic operator acting in the real Fock space  which represents a
Poincaré transformation, see Appendix.
$\bU(\Lambda,a)$ acts on the scalar quantum field as $$\bU(\Lambda,a)\,\bPhi(x)\,U^{-1}(\Lambda,a)
 =
 \bPhi(\Lambda x+a).
$$
\item Causality. In functional form it is
written as
\[
\frac{\delta}{\delta g(x)}
\left(
\frac{\delta \bS[g]}{\delta g(y)}
\bS^T[g]
\right)
=0,
\qquad
x\lesssim y .
\]
Here \(x\lesssim y\) means that \(x\) is not later than 
\(y\), i.e. \(x\) lies
in the causal past of \(y\), or is spacelike separated from \(y\).
\item  Unitarity.
For real \(g(x)\), the \(S\)-matrix must be orthogonal and symplectic:
\bea
\bS^T[g]\bS[g]
&=&
\bS[g]\bS^T[g]
=
\mathbf{1}.\\
\bS[g]\bJ&=&\bJ \bS[g]
\eea
\end{itemize}

\subsection{Orthogonality  and Symplectivity Conditions on $\bS$ Matrix in $\bJ$-QFT }
For $g(x)=1$ we write the $\bS$-matrix 
as
\be\label{SJ}
\bS=\id+\bJ \bT
\ee
The symplecticity condition means
\be \label{ST}\bJ \bT=\bT \bJ \ee
and the orthogonality condition 
\be
\label{UJ}
\bS^T\bS=\id 
\ee
is equivalent to
\be\label{UnT}
 \bJ(\bT^T-\bT)=\bT^T \bT\ee
 This relation is the analog of the usual unitarity relation, that one writes using the decomposition 
\be
S =1 + iT,
\ee
where $1$ describes the absence of scattering and
$T$ is the nontrivial transition operator.
Substituting this decomposition into
\[
S^\dagger S=\mathbf{1},
\]
one obtains
\be
i(T^\dagger-T)
=
T^\dagger T,
\ee
that is often is written as 
\be\label{1T-TT}
2\,\mIm\, T
=
T^\dagger T,
\ee
where $2\,\mIm \,T=i(T^\dagger-T)$.

Introducing notations
\be
\bT=T_\id \id + T_\bJ\,\bJ\ee
and assuming $T_\id^T=T_\id$,
$T_\bJ ^T=T_\bJ $, i.e.
\be
\bT^T=T^T_\id \id -\bJ \,T^T_\bJ\ee
we get 
\be
\bJ(\bT^T-\bT)=2 \,T_\bJ\ee
and one gets 
\be
 2\, T_\bJ=\bT^T \bT,\ee
that is an analog of \eqref{1T-TT}.
\\

Taking matrix elements between states \(|i\rangle\) and \(|f\rangle\), one obtains
\be\label{bT-TT}
 \bJ(\langle f|\bT^T|i\rangle-\langle f|\bT |i\rangle)=\sum_n
\langle f|\bT^T |n\rangle\langle n|\bT|i\rangle\ee
where \(\sum_n\) denotes a sum, and in general also an integral, over all
intermediate physical states.
As in the usual case we can rewrite the relation \eqref{bT-TT} as a relation  for the invariant amplitudes $\bM(i\to f)$\footnote{As usual, the invariant amplitude, also called the reduced matrix element is obtained from $\langle f|\bT |i\rangle$ 
after removing the trivial conservation law, more explicitly,
$
\bT_{fi}
=
(2\pi)^4
\delta^{(4)}(P_f-P_i)\,
\bM_{fi}.
$} that are now the 2$\times$2 matrices
\be \bM(i\to f)=\bM_\id(i\to f)\id+\bM_\bJ(i\to f)\bJ\ee
\be
\label{JUM}
2\bM_\bJ(i\to f)
=
\sum_n \int d\Pi_n \,(2\pi)^4\,\delta^{(4)}(p_i - p_n)\;
\mathbb{M}(i\to n)\,\mathbb{M}(f\to n)^T,
\ee
where the phase-space integral is
$
d\Pi_n = \prod_{j\in n} \frac{d^3 p_j}{(2\pi)^3\,2\omega_{p_j}}$.
\\

Let us check the relation \eqref{JUM}
on  simple examples. Consider   $\lambda 
\bPhi^4$
 theory, 
 \be L_{int}(\bPhi(x))=\lambda 
\bPhi^4.
\ee
The $\bS$ matrix is defined by the analog of the the Dyson formula
\be \bS=T\exp \{\,\bJ \int d^4x L_{int}(\bPhi(x))\}\ee
Let us consider  the scattering $2\to 2$, 
$|\bk_1;\bk_2\rangle\to
|\bk_1';\bk_2'\rangle$, where 
\be
|\bk_1;\bk_2\rangle=a_1^\dagger(\bk_1)a_2^\dagger(\bk_2)|0\rangle\ee

At first‑order, the invariant amplitude is simply the tree‑level vertex:
\be
\bT^{(1)}= g,\quad \bT^{T(1)}=g, \quad \bT^{T(1)}- \bT^{(1)} = 0,\qquad (\bT^T\bT)^{(1)}=0,
\ee
i.e. the relation \eqref{UnT} holds trivially.
\\

Let us now consider the second order on $g$. The Born approximation in mass-shell is zero by the kinematical reason.  Let us consider now the second order of the amplitude $|\bk_1;\bk_2\rangle\to
|\bk_1';\bk_2'\rangle$. 

\newpage
\subsection{One‑Loop  Amplitude $\bM_{1-loop}$}

In $\lambda 
\bPhi^4$
 theory, the bubble one‑loop  correction to the four‑point vertex is

\bea\label{JbM1}
\bJ\,\bM_{1-loop}(s) &=& \frac{\lambda^2}{2} \bI(s), \\\label{JI} \bI(s)&=&-\int \frac{d^4 k}{(2\pi)^4} \tilde\bG^{\bJ}_{F,\epsilon}(k) \bG^{\bJ}_{F,\epsilon}(p-k),\quad s \equiv p^2, \eea
where $\tilde\bG^{\bJ}_{F,\epsilon}(k)$ is defined by \eqref{bGeps}. 
The explicit form of $\bI(s)$ is given by \bea \bI(s) &=&  \frac{\bJ}{16\pi^2} \left[ \frac{1}{\hat{\epsilon}} + \ln\left(\frac{\mu^2}{m^2}\right) - \bF(s) \right], \\ \bF(s) &=& \begin{cases} F_<(s), \quad & 0<s < 4m^2,\\[8pt] F_>(s) - \bJ\pi \sqrt{1 - \frac{4m^2}{s}}, & s > 4m^2, \end{cases} \eea 
where \bea F_<(s)&=&-2 + 2\arctan\Big(\frac{1}{\sqrt{\frac{4m^2}{s}-1}}\Big),\\ F_>(s)&=& -2 + \sqrt{1 - \frac{4m^2}{s}} \left[\ln\!\left(\frac{1 + \sqrt{1 - \frac{4m^2}{s}}}{1 - \sqrt{1 - \frac{ 4m^2}{s}}}\right) \right], \eea
and
 \be \label{epsilon-reg}\frac{1}{\hat{\epsilon}} = \frac{1}{\epsilon} - \gamma_E + \ln(4\pi).\ee 
The form of \eqref{epsilon-reg} is related to dimensional regularization $d = 4 - 2\epsilon$ and the MS subtraction scheme.
\\

For comparison, in the usual complex formulation the analogue of  \eqref{JbM1} is 
\bea \label{usualM} 
iM_{1\text{-loop}}(s) &=& \frac{\lambda^2}{2}I(s)\\
\label{Ip} I(s) &=& \int \frac{d^4 k}{(2\pi)^4} \frac{1}{k^2 - m^2 + i\epsilon} \frac{1}{(p-k)^2 - m^2 + i\epsilon}, \quad s \equiv p^2. \eea 
The integral in \eqref{Ip} is logarithmically divergent in the ultraviolet and must be regularised. Using dimensional regularisation $d = 4 - 2\epsilon$ and the $\overline{\mbox{MS}}$ subtraction scheme, the explicit evaluation gives the well known answer\footnote{For different masses one gets the Passarino–Veltman scalar two-point function $B_0(p^2;m_1^2,m_2^2)$ that is one of the most common loop integrals in QFT \cite{Passarino:1978jh,Siringo:2012mi}.} 
\bea I(s) &=& \frac{i}{16\pi^2} \left[ \frac{1}{\hat{\epsilon}} + \ln\left(\frac{\mu^2}{m^2}\right) - F(s) \right], \\ F(s) &=& \begin{cases} F_<(s), \quad & s < 4m^2,\\[8pt] F_>(s) - i\pi \sqrt{1 - \frac{4m^2}{s}}, & s > 4m^2. \end{cases} \eea 
 
For \(s > 4m^2\), the two internal propagators in \eqref{JI} as well in \eqref{Ip} can simultaneously go on shell. 
The $\bJ$-part of the loop integral is : 
\be\label{bM1} \bM_{\bJ,1-loop}(s) = \frac{\lambda^2}{32\pi} \sqrt{1 - \frac{4m^2}{s}}. \ee It is the analog of imaginary part of the loop integral: 
\be\label{Mim} \mIm\, M_{1-loop}(s) = \frac{\lambda^2}{32\pi} \sqrt{1 - \frac{4m^2}{s}}. \ee 
The relation \eqref{bM1} expresses the validity of the $\bJ$-optical theorem for the forward \(2 \to 2\) amplitude. Indeed,  in the 1-loop approximation this theorem reads 
\be\label{ROT} 
2 \bM_{\bJ,1-loop}(s)=\frac12\int d\Phi_2 \,\tr[\,\bM^T_{tree}(s)\bM_{tree}(s)] 
\ee 
At tree level, $ \bM_{tree} = -\lambda \id$ and therefore $\frac12\tr[\,\bM^T_{tree}(s)\bM_{tree}(s)]=\lambda^2$. Thus \eqref{ROT} becomes
\be\label{bM1e} 
2 \bM_{\bJ,1-loop}(s)=\lambda^2\int d\Phi_2. 
\ee 
The two‑body phase space for identical final particles is 
\be\label{dphi} 
\int d\Phi_2 = \frac{1}{2!} \frac{d^3 k_1}{(2\pi)^3 2\omega_{k_1}} \frac{d^3 k_2}{(2\pi)^3 2\omega_{k_2}} (2\pi)^4 \delta^{(4)}(p - k_1 - k_2)= \frac{1}{16\pi} \sqrt{1 - \frac{4m^2}{s}}. 
\ee 
Substituting \eqref{dphi} into \eqref{bM1e}, we find
\be \bM_{\bJ,1\text{-loop}}(s)
=
\frac{\lambda^2}{32\pi}
\sqrt{1-\frac{4m^2}{s}},
\ee
which coincides  as \eqref{bM1}. \\

This is the direct analogue of the usual one-loop optical theorem
\cite{MS},
\begin{equation}
\label{eq:usual-optical}
2\,\operatorname{Im} M_{1\text{-loop}}(s)
=\int d\Phi_2\,
\left|M_{\text{tree}}(s)\right|^2 .
\end{equation}
Indeed, using
$
M_{\text{tree}}(s)=-\lambda
$
and the phase-space integral \eqref{dphi}, one obtains
$$
\mIm\, M_{1\text{-loop}}(s)
=
\frac{\lambda^2}{32\pi}
\sqrt{1-\frac{4m^2}{s}},
$$
which is precisely \eqref{Mim}.

\newpage

\subsection{Decays in RQFT}
In the theory with an interaction action $$V(\{\phi_i\})=g\int \phi_1 (x)\phi_2(x) \phi_3(x),  $$
where $\phi_i$ are scalar fields with masses $m_1,m_2,m_3$ under condition 
\be 
m_{3} > m_1 + m_2,
\ee
the particle 3 can decay on particles 1 and 2. 

To be more concrete, let us consider 
 a classic example of a spontaneous decay where one parent particle decays into exactly two daughter particles - the decay of the short-lived neutral kaon to pions:
\[
K^0_S \;\longrightarrow\; \pi^+ + \pi^-.
\]

The  rest masses are: 
    $
    m_{K} \approx 498$ MeV,
$m_{\pi^+} = m_{\pi^-}\approx 140$ MeV.
  The decay is allowed since:
$
m_K > m_{\pi^+} + m_{\pi^-}.
$

The differential decay rate (probability of decay per unit time) for a $1 \to 2$ process is given by
\be\label{dGamma}
d\Gamma = \frac{\tr(\bM^T\,\bM) }{2 m_K}
(2\pi)^4 \delta^{(4)}(p - k_1 - k_2)
\frac{d^3 \bk_1}{(2\pi)^3 \, 2\omega_{\bk_1}}
\frac{d^3 \bk_2}{(2\pi)^3 \, 2\omega_{\bk_2}},
\ee
with $\omega_\bk = \sqrt{|\bk|^2 + m_\pi^2}$.

The total decay width is obtained by integrating over the two‑body phase space:
\be
\Gamma(K \to \pi^+ \pi^-) =\frac{g^2}{2 m_K} \int d\Phi_2=
\frac{g^2}{16\pi m_K^2} \sqrt{m_K^2 - 4m_\pi^2}
\ee
Note that here all calculations use only real number and $\Gamma$ as well as the decay probability $P(t)=
P(t) = 1 - e^{-\Gamma t}
$. Therefore,  RQFT they are the same as in usual case. 
\\
\newpage
\section{Realification  of Standard Model}\label{RSM}
Here we choose the Standard Model of elementary particles to investigate the feasibility of purely real formulations. This choice is motivated by the theory's unique status as a realistic model of fundamental interactions, which merits comprehensive study from both phenomenological and conceptual perspectives.

\subsection{Realification of  Fermions  of Standard Model}\label{app:GSM-fermi}
The fermionic sector of the Standard Model is  formulated in terms of
complex chiral Weyl fields \cite{MS,DTong}. In a pure real formulation, each complex Weyl field
is replaced by a real field with twice as many real components, together with a
real complex-structure operator $\bJ$, see Appendix \ref{app:RealWeyl} about realification complex Weyl fermion fields.
The realification does not change the physical particle content. It only
replaces complex coordinates by pairs of real coordinates.
\\

For one fermion generation, the left-handed quark doublet is
\begin{equation}
Q_L=
\begin{pmatrix}
u_L\\
d_L
\end{pmatrix}
\sim
(\mathbf{3},\mathbf{2},\tfrac16).
\end{equation}
and left-handed lepton doublet is
\begin{equation}
L_L=
\begin{pmatrix}
\nu_L\\
e_L
\end{pmatrix}
\sim
(\mathbf{1},\mathbf{2},-\tfrac12).
\end{equation}
The right-handed fields are
\begin{equation}
u_R\sim(\mathbf{3},\mathbf{1},\tfrac23),\quad
d_R\sim(\mathbf{3},\mathbf{1},-\tfrac13),
\quad
e_R\sim(\mathbf{1},\mathbf{1},-1).
\end{equation}
In the minimal Standard Model there is no right-handed neutrino.
The complete fermion content is obtained by repeating this set for three
generations. The realification of a Weyl fermion in the Standard Model is carried out along the lines of the discussion in Section \ref{app:RealWeyl}, specifically equations \eqref{W1} and \eqref{W2}.
The realified fermion fields are denoted by
$$
\mathbb{Q}_L
=
\mathcal{R}(Q_L),
\quad
\mathbb{L}_L
=
\mathcal{R}(L_L),
\quad
\bu_R
=
\mathcal{R}(u_R),\quad
\bd_R
=
\mathcal{R}(d_R),\quad
\mathbb{e}_R
=
\mathcal{R}(e_R).$$
For example, the complex quark doublet $Q_L$ has
$
2\times 3\times 2=12$
complex components when the Weyl-spinor, color, and weak-isospin indices are
included. Its realification $\mathbb{Q}_L$ has
24
real Grassmann components. Realification does not double the physical content. It merely replaces twelve complex
components by twenty-four real components.

The fermion content consists of three generations:
1-st generation:
$
u,\, d,\, e,\, \nu_e$,
second generation: $
c,\, s,\, \mu,\, \nu_\mu$
and third generation:
$
t,\, b,\, \tau,\, \nu_\tau$. The realification of the second and third generations proceeds analogously to the first generation. The gauge interactions are identical for all three generations. In what follows, we consider only the first generation in detail.

\subsection{Gauge Transformation of the Real Fermion Fields}\label{RSM-gauge}

The Standard Model gauge group remains $G_{\mathrm{SM}}=
SU(3)_c\times SU(2)_L\times U(1)_Y,$ but its complex unitary representations are rewritten in realified formulation as real representations
that are simultaneously orthogonal and symplectic.

In the realified formulation, each complex gauge transformation
\bea
&&U(x)
=
U_3(x)\,
U_2(x)\,
e^{iY\alpha_Y(x)},\\
&& U_3(x)\in SU(3)_c,\quad
U_2(x)\in SU(2)_L,\quad e^{iY\alpha_Y(x)}\in U(1)_Y,
\eea
is replaced by its realification
\begin{equation}
\mathcal{R}(U)
=
\begin{pmatrix}
\operatorname{Re}U & -\operatorname{Im}U \\
\operatorname{Im}U & \operatorname{Re}U
\end{pmatrix}.
\end{equation}
On the realified Weyl field \eqref{W2} the finite gauge transformation acts as
\be
\Xi(x)
\longrightarrow
\mathcal{R}
\left(
U_3(x)\,
U_2(x)\,
e^{iY\alpha_Y(x)}
\right)
\Xi(x).
\ee
Since the realification map preserves matrix products,
$
\mathcal{R}(U_1U_2)
=
\mathcal{R}(U_1)\,
\mathcal{R}(U_2),
$
the transformation may equivalently be written as
\begin{equation}
\Xi(x)
\longrightarrow
\mathcal{R}(U_3(x))\,
\mathcal{R}(U_2(x))\,
\mathcal{R}
\left(
e^{iY\alpha_Y(x)}
\right)
\Xi(x).
\end{equation}
More details about realification of group $U(1)$, $SU(2)$ and $SU(2)$ see Appendix \ref{app:GSM}.
\\

For an infinitesimal gauge transformation,
\begin{equation}
U_3(x)
=
I+i\alpha_c^a(x)T_c^a
+
\mathcal{O}(\alpha_c^2),
\end{equation}
\begin{equation}
U_2(x)
=
I+i\alpha_L^i(x)T_L^i
+
\mathcal{O}(\alpha_L^2),
\end{equation}
and
\begin{equation}
e^{iY\alpha_Y(x)}
=
I+iY\alpha_Y(x)
+
\mathcal{O}(\alpha_Y^2).
\end{equation}

Therefore, the infinitesimal transformation of the realified Weyl field is
\be
\delta\Xi
=
\left[
\alpha_c^a\,
\mathcal{R}(iT_c^a)
+
\alpha_L^i\,
\mathcal{R}(iT_L^i)
+
Y\alpha_Y\,\bJ
\right]\Xi.
\ee

It is convenient to introduce the real generators
\begin{equation}
\mathbb{T}_c^a
=
\mathcal{R}(iT_c^a),
\qquad
\mathbb{T}_L^i
=
\mathcal{R}(iT_L^i).
\end{equation}

Then
\be
\delta\Xi
=
\left[
\alpha_c^a\mathbb{T}_c^a
+
\alpha_L^i\mathbb{T}_L^i
+
Y\alpha_Y\bJ
\right]\Xi.
\ee

If the original generators $T_c^a$ and $T_L^i$ are Hermitian, then the
realified generators are real antisymmetric matrices:
\begin{equation}
\left(\mathbb{T}_c^a\right)^T
=
-\mathbb{T}_c^a,
\qquad
\left(\mathbb{T}_L^i\right)^T
=
-\mathbb{T}_L^i.
\end{equation}

They also commute with the complex structure:
\begin{equation}
\left[
\mathbb{T}_c^a,\bJ
\right]
=
0,
\qquad
\left[
\mathbb{T}_L^i,\bJ
\right]
=
0.
\end{equation}

Consequently, the realified gauge transformations preserve both the Euclidean
metric and the symplectic structure defined by $\bJ$. In particular,
\begin{equation}
\mathcal{R}(U)^T\mathcal{R}(U)
=
I,
\end{equation}
and
\begin{equation}
\mathcal{R}(U)^T\bJ\mathcal{R}(U)
=
\bJ.
\end{equation}

Hence the realified Standard Model gauge group acts through matrices belonging
to
\begin{equation}
O(2N)\cap Sp(2N,\mathbb{R}),
\end{equation}
with the additional condition that these matrices commute with the fixed
complex structure $\bJ$.
Therefore, the gauge group itself remains
$
SU(3)_c
\times
SU(2)_L
\times
U(1)_Y,
$
but it acts on the realified Weyl fermions through real
orthogonal-symplectic matrices. The complex phases of the ordinary formulation
are represented by real rotations generated by the complex-structure operator
$\bJ$.
\\

\subsubsection*{Gauge Transformations of the realified Weyl fermions of one generation}
The realified  left-handed quark quadruplet
\begin{equation}\label{trnsf}
\bQ_L
=
\begin{pmatrix}
\bu_L\\
\bd_L
\end{pmatrix}\quad \mbox{transforms as}\quad
\bQ_L
\longrightarrow
\bU_c\,
\bU_L\,
e^{\bJ\alpha_Y/6}
\bQ_L.
\ee
With explicit color and weak-isospin indices, \eqref{trnsf} means 
\be
(\bQ_L)^{\bba\br}
\longrightarrow
(\bU_c)^{\bba}_{\bb}
(\bU_L)^{\br}_{\bs}
e^{\bJ\alpha_Y/6}
(\bQ_L)^{\bb\bs},\quad \mbox{where}\quad
\bba,\,\bb= \overline{1,6},
\qquad
\br,\bs=1,2.
\ee
The  relified right-handed up and down quark transforms as
\be
\bu_R
\longrightarrow
\bU_c\,
e^{2\bJ\alpha_Y/3}
\bu_R, \qquad
\bd_R
\longrightarrow
U_c\,
e^{-\bJ\alpha_Y/3}
\bd_R.
\ee
The  relified left-handed lepton doublet is
\begin{equation}
\bL_L
=
\begin{pmatrix}
\bnu_L\\
\bbe_L
\end{pmatrix}\quad \mbox{transforms as}\quad
\bL_L
\longrightarrow
\bU_L\,
e^{-\bJ\alpha_Y/2}
\bL_L.
\end{equation}
The  relified  right-handed charged lepton transforms as
\be
\bbe_R
\longrightarrow
e^{-\bJ\alpha_Y}
e_R.
\ee

All these transformations are real, orthogonal, symplectic, and commute with
$\bJ$.

\subsection{Gauge Curvature in the Real Formulation}
The curvature of the real covariant derivative is defined by
$
[\mathbb{D}_\mu,\mathbb{D}_\nu]
=
\mathbb{F}_{\mu\nu}$.
It takes the form
\begin{equation}
\mathbb{F}_{\mu\nu}
=
g_sG_{\mu\nu}^A\mathbb{T}_C^A
+
gW_{\mu\nu}^I\mathbb{T}_L^I
+
g'YB_{\mu\nu}\bJ.
\end{equation}

Here
\begin{equation}
G_{\mu\nu}^A
=
\partial_\mu G_\nu^A
-
\partial_\nu G_\mu^A
+
g_s f^{ABC}G_\mu^B G_\nu^C,
\end{equation}
\begin{equation}
W_{\mu\nu}^I
=
\partial_\mu W_\nu^I
-
\partial_\nu W_\mu^I
+
g\epsilon^{IJK}W_\mu^\bJ W_\nu^K,
\end{equation}
and
\begin{equation}
B_{\mu\nu}
=
\partial_\mu B_\nu
-
\partial_\nu B_\mu.
\end{equation}

All matrices entering $\mathbb{F}_{\mu\nu}$ are real and antisymmetric.

\subsection{Realification of the Higgs Field}
\label{Higgs}
The Standard Model contains one complex scalar doublet,
\be
H(x)
=
\begin{pmatrix}
H^+(x)\\
H^0(x)
\end{pmatrix}.
\ee
It is a singlet under the color group $SU(3)_c$, a doublet
under $SU(2)_L$, and has weak hypercharge
$
Y_H=\frac12$,
\be
H(x)
\longrightarrow
U_L(x)\,
e^{i\alpha_Y(x)/2}
H(x).
\ee
The gauge-invariant kinetic term is
$
\mathcal L_{\mathrm{kin}}
=
(D_\mu H)^\dagger
D^\mu H,
$
where
the gauge-covariant derivative acting on the Higgs doublet is
\be
D_\mu
=
\partial_\mu
-
igW_\mu^i\frac{\tau_i}{2}
-
ig'\frac12B_\mu,
\ee
here $W_\mu^i$ are the $SU(2)_L$ gauge fields, $B_\mu$ is the $U(1)_Y$ gauge
field, and $\tau_i$ are the Pauli matrices.
The  Higgs potential is
$
V(H)
=
-\mu^2H^\dagger H
+
\lambda(H^\dagger H)^2$. The Standard Model fermions interact with the Higgs field through the Yukawa
Lagrangian,
\begin{equation}\label{LY}
\mathcal L_Y
=
-
\overline Q_LY_dHd_R
-
\overline Q_LY_u\widetilde H u_R
-
\overline L_LY_eHe_R
+
\mathrm{h.c.},
\end{equation}
here $
\widetilde H
=
i\sigma^2H^*$
is the conjugate Higgs doublet, and $Y_u$, $Y_d$, and $Y_e$ are complex Yukawa
matrices in generation space. Thus the presence of 
$\tilde H$
 only in the one Yukawa term  is the unique choice that makes all three Yukawa interactions invariant under
$SU_L(2)\times U(1)$.

In all-left-handed two-component notation, the Standard Model Yukawa
Lagrangian is
\bea
\mathcal L_Y
&=&
-\mathcal Y_u
-\mathcal Y_d
-\mathcal Y_e
+\mathrm{h.c.},
\\
\mathcal Y_u
&=&
(Y_u)_{ij}
\epsilon_{\alpha\beta}
\epsilon_{ab}
(u_L^c)_{i}^{\alpha A}
(Q_L)_{jA}^{\beta a}
H^b,
\label{eq:complex-up-yukawa}
\\
\mathcal Y_d
&=&
(Y_d)_{ij}
\epsilon_{\alpha\beta}
(d_L^c)_{i}^{\alpha A}
(Q_L)_{jA}^{\beta a}
H_a^*,
\label{eq:complex-down-yukawa}
\\
\mathcal Y_e
&=&
(Y_e)_{ij}
\epsilon_{\alpha\beta}
(e_L^c)_{i}^{\alpha}
(L_L)_{j}^{\beta a}
H_a^*.
\label{eq:complex-electron-yukawa}
\eea
The hypercharges in these terms are $0$.
\\

Realification of the Standard Model complex Higgs doublet
$H\in\mathbb C^2$
is obtained by  replacing complex Higgs doublet by
\begin{equation}
\mathbb H
=
\begin{pmatrix}
\operatorname{Re}H\\
\operatorname{Im}H
\end{pmatrix}
\in\mathbb R^4.
\end{equation}
Under the electroweak gauge transformation,
$
H
\longrightarrow
U_Le^{i\alpha_Y/2}H,
$
the realified Higgs field transforms as
\be
\bH
\longrightarrow
\cR\left(U_Le^{i\alpha_Y/2}\right)\bH.
\ee
The  complex Yukawa matrix $Y$ in \eqref{LY} can be realified accoding the generel formula:
\begin{equation}
Y=A+iB
\quad\longrightarrow\quad
\mathcal R(Y)
=
\begin{pmatrix}
A&-B\\
B&A
\end{pmatrix}.
\end{equation}
The realification of the conjugate Higgs doublet is
\be
\widetilde{\mathbb H}
=
\mathcal R(i\sigma^2)\,
\mathbb K\,
\mathbb H
,\quad \mbox{where}\quad 
\mathbb K
=
\begin{pmatrix}
I_2&0\\
0&-I_2
\end{pmatrix}.\ee
The realified Yukawa sector becomes a real multilinear coupling among
\begin{equation}
\bQ_{L},
\quad
\bu_{R},
\quad
\bd_{R},
\quad
\bL_{L},
\quad
\bbe_{R},
\quad
\bH.
\end{equation}
and after algebraic calculations
\bea\nn
\fL_Y\equiv\cR(\cL_Y)
&=&-2(\bT_{u})_{ij;rst}
\epsilon_{\alpha\beta}
\epsilon_{ab}
(\mathbb u_L^c)_{i,r}^{\alpha A}
(\mathbb Q_L)_{jA,s}^{\beta a}
\mathbb H_t^b
-2
(\bT_d)_{ij;rst}
\epsilon_{\alpha\beta}
(\mathbb d_L^c)_{i,r}^{\alpha A}
(\mathbb Q_L)_{jA,s}^{\beta a}
\mathbb H_{a,t}
\\
&-&2(\bT_{e})_{ij;rst}
\epsilon_{\alpha\beta}
(\mathbb e_L^c)_{i,r}^{\alpha}
(\mathbb L_L)_{j,s}^{\beta a}
\mathbb H_{a,t}.
\eea
where
\be
\begin{aligned}
(\bT_u)_{ij;rst}&=
(\mRe Y_{u})_{ij}C_{rst}
-
(\mIm Y_{u})_{ij}S_{rst}
\\
(\bT_d)_{ij;rst}&=(-1)^t
\left[
(\mRe Y_{d})_{ij}C_{rst}
-
(\mIm Y_{d})_{ij}S_{rst}
\right]
\\
(\bT_{e})_{ij;rst}&=
(-1)^t
\left[
(\mRe Y_{e})_{ij}C_{rst}
-
(\mIm Y_{e})_{ij}S_{rst}
\right]
.
\end{aligned}
\end{equation}
and
\begin{equation}
C_{rst}
=
\cos\left(
\frac{\pi}{2}(r+s+t)
\right),\qquad
S_{rst}
=
\sin\left(
\frac{\pi}{2}(r+s+t)
\right).
\end{equation}
Complex phases in the Yukawa matrices are not lost. They are encoded in the
parts of the realified matrices that mix the real and imaginary components.
Consequently, the CKM (Cabibbo–Kobayashi–Maskawa) and PMNS (Pontecorvo–Maki–Nakagawa–Sakata) phases remain present.

\subsection{Spontaneous Symmetry Breaking}

The realification of the Higgs doublet
\be
H
=
\frac{1}{\sqrt{2}}
\begin{pmatrix}
0\\
v+h
\end{pmatrix}.
\ee
where $h$ is the quantum field and $v$ is a classical constant, after Spontaneous symmetry breaking in the unitary gauge,  is
\be
\bH
=
\frac{1}{\sqrt{2}}
\begin{pmatrix}
0\\
v+h\\
0\\
0
\end{pmatrix},
\end{equation}
in the basis
\begin{equation}
\mathbb H
=
\begin{pmatrix}
\operatorname{Re}H^+\\
\operatorname{Re}H^0\\
\operatorname{Im}H^+\\
\operatorname{Im}H^0
\end{pmatrix}.
\end{equation}

The complex fermion mass matrices are
\begin{equation}
M_f
=
\frac{v}{\sqrt{2}}Y_f,
\qquad
f=u,d,e.
\end{equation}

Their realifications are
\begin{equation}
\mathbb M_f
=
\mathcal R(M_f)
=
\frac{v}{\sqrt{2}}
\begin{pmatrix}
Y_{fR}&-Y_{fI}\\
Y_{fI}&Y_{fR}
\end{pmatrix}.
\end{equation}

Replacing \(v\) by \(v+h\), one obtains both the mass term and the interaction
with the physical Higgs boson:
\begin{equation}
\mathbb M_f(h)
=
\frac{v+h}{\sqrt{2}}\,
\mathbb Y_f.
\end{equation}

Therefore,
the Higgs coupling to each fermion is consequently still proportional to the
corresponding fermion mass.
\\

To conclude this section, we note that the purely real formulation of the Standard Model is not a new theory, but a more unified mathematical representation of the existing theory. One can say that it provides a more powerful perspective for understanding the ultimate laws of nature and serves as an indispensable tool for exploring physics beyond the Standard Model, particularly in the realms of unification, neutrino physics, and non-perturbative quantum field theory.

An important open question for future work is whether the isomorphism
$U(N)
\cong
O(2N)\cap Sp(2N,\bR)
$
reveals deeper geometric structures inherent to the Standard Model that are inaccessible through ordinary complex formulation.
\subsection{Breaking the \(\mathbb{J}\)-Symmetry and Mass Terms in Realified Models}

It is well known that in the Weyl spinor formulation, the Standard Model contains no explicit mass terms for fermions in its fundamental Lagrangian—that is, no gauge-invariant bare fermion mass terms before electroweak symmetry breaking. All masses must be generated through the Higgs mechanism, with the sole exception of a possible Majorana mass for a right-handed neutrino singlet, which is not part of the minimal Standard Model. The \(\mathbb{J}\)-symmetry serves as a "protector" of masslessness in the realified Standard Model.

The violation of the \(\mathbb{J}\)-symmetry in the realified formalism allows for the presence of mass terms for fermions. Therefore, when this symmetry is broken, it opens the door to a variety of fermion mass terms that are either forbidden or require new physics in the complex formulation.

In general, in the realified formalism, the complex structure of quantum field theory is encoded in a real matrix \(\mathbb{J}\) satisfying \(\mathbb{J}^2 = -\mathbb{I}\). The \(\mathbb{J}\)-symmetry is the invariance of the Lagrangian under global \(\mathbb{J}\)-rotations:
\begin{equation}
\Xi \mapsto e^{\theta \mathbb{J}} \Xi, \label{eq:J_rotation}
\end{equation}
which corresponds to the usual \(U(1)\) phase symmetry \(\chi \mapsto e^{i\theta} \chi\) in the complex formulation. Let us note that the presence of an overall phase factor in quantum mechanics also generates a transformation analogous to \eqref{eq:J_rotation} and is related to the emergent property of spectral degeneracy in quantum mechanics \cite{Arefeva:2025zbx}.

The question we address here is: assume the \(\mathbb{J}\)-symmetry is broken. Are mass terms for fermions available in this case? The answer is positive, but with important nuances. In the realified Standard Model, the following mass terms are enabled by \(\mathbb{J}\)-symmetry breaking:

\begin{itemize}

\item The dimension-five Weinberg operator. The operator
\begin{equation}
\mathcal{L}_5 = -\frac{1}{2\Lambda} \kappa_{ij} (L_i H)(L_j H) + \text{h.c.} \label{eq:Weinberg_complex}
\end{equation}
 In the realified formalism, it becomes:
\begin{equation}
\mathcal{L}_5 = -\frac{1}{2\Lambda} \kappa_{ij} (\mathbb{L}_i^T \mathcal{C} \mathbb{H}) (\mathbb{L}_j^T \mathcal{C} \mathbb{H}). \label{eq:Weinberg_real}
\end{equation}
It is gauge-invariant but explicitly breaks the \(\mathbb{J}\)-symmetry.
After electroweak symmetry breaking, it generates a Majorana mass term for the left-handed neutrinos:
\begin{equation}
\mathcal{L}_{M,\nu} = -\frac{1}{2} (m_\nu)_{ij} \, \Xi_{\nu_i}^T \mathcal{C} \Xi_{\nu_j}, \qquad (m_\nu)_{ij} = \frac{v^2}{2\Lambda} \kappa_{ij}. \label{eq:neutrino_mass_real}
\end{equation}
Thus, \(\mathbb{J}\)-symmetry breaking together with electroweak symmetry breaking yields a Majorana mass for \(\nu_L\).

\item  In the seesaw mechanism, we add gauge-singlet right-handed neutrinos \(N_i \sim (1,1,0)\). In the realified formalism, these become real four-component spinors \(\Xi_{N_i}\). The mass terms are:
\begin{equation}
\mathcal{L}_{\text{mass}} = -\frac{1}{2} M_R \, \Xi_N^T \mathcal{C} \, \Xi_N - m_D \, \Xi_L^T \mathcal{C} \, \Xi_N  \label{eq:seesaw_real}
\end{equation}
Here $\Xi_L$  is the realified spinor for the left-handed neutrino $\nu_L$ and $\Xi_N$ is the realified spinor for the right-handed neutrino singlet $N$ (often denoted \(\nu_R\)).
The Dirac mass term \(m_D \Xi_L^T \mathcal{C} \Xi_N\) preserves the \(\mathbb{J}\)-symmetry if \(\Xi_L\) and \(\Xi_N\) transform in the same way.
The Majorana mass term \(\frac{1}{2} M_R \Xi_N^T \mathcal{C} \Xi_N\) breaks the \(\mathbb{J}\)-symmetry.

If \(\mathbb{J}\)-symmetry is broken, the Majorana mass term for the right-handed neutrino is allowed. This leads to the seesaw mechanism:
\begin{equation}
m_\nu \sim \frac{m_D^2}{M_R}. \label{eq:seesaw_mass}
\end{equation}
Thus, \(\mathbb{J}\)-symmetry breaking together with the seesaw mechanism leads to light Majorana neutrino masses.

\end{itemize}

Thus, \(\mathbb{J}\)-symmetry breaking is a necessary condition for Majorana mass terms in the realified formalism, but it must be combined with gauge invariance to produce physically consistent mass terms. This makes \(\mathbb{J}\)-symmetry a powerful diagnostic tool for identifying new physics in the realified Standard Model.

\newpage
\section{Conclusion: Discussion and Future Directions}

In this paper, the question of generalizing the quantum mechanical approach—which is formulated for a finite number of degrees of freedom and deals only with real numbers—to the infinite degrees of freedom case in quantum field theory is investigated. This question has two aspects. 
\begin{itemize}
\item The first concerns the quantization of a real scalar field: specifically, can we construct such a quantization without referring to $i$? In this case, the quantization generalizes \eqref{JQM} and is given by $\bJ$-quantization \eqref{CRJ}. 

\item The second aspect asks whether the theory of elementary particles can be reformulated or extended using only real fields.
\end{itemize}

As the answer to the first question, in this paper we have presented a formulation of scalar quantum field theory
entirely in terms of real structures. We call this theory RQFT. The main idea is to replace the external
complex unit \(i\) by an internal real matrix $\bJ$, satisfying
$
\bJ^2=-\id
$,
which plays the role of the complex structure on a real Kähler space. In this
way the usual complex Hilbert-space formulation of quantum field theory is
rewritten as a theory on a real Hilbert space equipped with a compatible symplectic structure.
\\

The construction of RQFT is equivalent to the standard complex formulation of QFT at the level
of physical observables. Its purpose is not to change the numerical predictions
of scalar QFT, but to formulate QFT entirely in real Kähler terms. This makes
the roles of the orthogonal, symplectic, and complex structures explicit and
provides a real-matrix version of perturbative tools such as the Feynman
propagator, the Sokhotski--Plemelj formula, the Cutkosky rules, and the
optical theorem.
\\

We have constructed explicitly the free scalar field operator
$
\hat\bPhi(t,\bx)
$
acting in the real Kähler analogue of the usual bosonic Fock space. This field
satisfies the Klein--Gordon equation and has the canonical commutation relation 
\be
[\hat\bPhi(t,\bx),\partial _t\hat\bPhi(t,\by)]
=
\bJ\,\delta^{(3)}(\mathbf x-\mathbf y).
\ee
Thus the ordinary factor \(i\) in the canonical commutation relations is
replaced by the real symplectic matrix $\bJ$. This shows that the
standard canonical structure of scalar quantum field theory can be formulated
without introducing complex numbers as fundamental quantity.
\\

We have also introduced the corresponding 
$\bJ$-Wightman functions.
For the two-point function we found the direct analogue of the standard
positive-frequency Wightman distribution,
\be
\bW_2(x;y;m)
=
\int \frac{d^3p}{(2\pi)^3\,2\omega_{\mathbf p}}\,
e^{-\bJ p\cdot(x-y)} .
\ee
For the massless theory this gives the distribution
\be
\bW_2(x;y;0)
=
\frac{1}{4\pi^2}
\frac{1}{-(x-y)^2+\bJ\,0\,(x^0-y^0)},
\ee
for the massive theory we obtained
\be
\bW_2(x;y;m)
=
\frac{m}{4\pi^2}
\frac{
K_1\!\left(m
\sqrt{-(x-y)^2+\bJ\,0\,(x^0-y^0)}\right)
}{
\sqrt{-(x-y)^2+\bJ\, 0\,(x^0-y^0)^0}
}.
\ee
The spacelike, timelike and light-cone regimes were described explicitly. In
particular, the leading light-cone singularity of the massive Wightman function
coincides with the massless Wightman distribution, while the full massive
function contains additional mass-dependent logarithmic terms. This agrees with
the standard local singularity structure of relativistic quantum fields.
\\

We have then constructed the 
\(\bJ\)-Feynman Green function and computed
its \(\bJ\)-Fourier transform. The resulting momentum-space propagator has
the form
\be
\tilde{\bG}_F^{\,\bJ}(p)
=
\bJ\,\cP\Big(\frac{1}{p^2-m^2}
\Big)+
\pi\delta(p^2-m^2).
\ee
This formula is the real Kähler analogue of the usual Feynman prescription. It
also gives a natural \(\bJ\)-version of the Sokhotski--Plemelj formula.
The imaginary part of the usual complex propagator is replaced by the
\(\bJ\)-part of the real matrix-valued propagator.
\\

A central point of the construction of RQFT is the reformulation of unitarity. In the
ordinary complex theory the scattering matrix satisfies
$
S^\dagger S=\mathbf I$.
In the real Kähler formulation this condition is rewritten as the simultaneous
orthogonality and symplecticity of the scattering operator,
\be
S^T S=\id,
\qquad
S^T\bJ S=\bJ .
\ee
For orthogonal operators the second condition is equivalent to $[S,\bJ]=0$.
Thus the usual unitary structure of quantum field theory is equivalently
described by an orthogonal-symplectic structure on the corresponding real
Kähler space.
\\

We have also discussed observables in real quantum field theory. The physical
observables, such as masses, decay rates, scattering probabilities and cross
sections, are real quantities and coincide with the corresponding quantities in
the usual complex formulation. Therefore the real Kähler formulation does not
change the numerical predictions of ordinary scalar quantum field theory.
Rather, it gives an equivalent description in which the complex structure is
made explicit as a real symplectic operator.
\\

As an illustration, we considered the one-loop contribution to the \(2\to2\)
scattering amplitude in \(\bPhi^4\) theory. For \(s>4m^2\), the two
internal particles can go on shell, and the \(\bJ\)-part of the loop
amplitude is
\be
\bM_{\bJ,1-loop}
(s)
=
\frac{\lambda^2}{32\pi}
\sqrt{1-\frac{4m^2}{s}} .
\ee
This is the exact analogue of the imaginary part of the usual one-loop
amplitude. The corresponding relation is naturally interpreted as the
$\bJ$-optical theorem. In this way, the Cutkosky rules also acquire a
real Kähler form: the discontinuity of a loop diagram is expressed through its
\(\bJ\)-part, while the cut propagators are replaced by the standard
on-shell phase-space factors.
\\

The results obtained in this paper show that perturbative scalar quantum field theory
can be formulated completely in real terms. The replacement
$ 
i\longmapsto \bJ
$
is not merely a formal rule, but reflects the passage from a complex Hilbert
space to an equivalent real Kähler space. In this formulation, the three
structures
$g,\, \omega,\, \bJ
$
respectively encode the real scalar product, the symplectic form and the
complex structure. Standard QFT is recovered after identifying the action of
\(\bJ\) with multiplication by \(i\).
\\

It would be natural to extend the present
construction to spinor and gauge fields, where internal symmetries and
constraints play an essential role. Another important problem is to develop the
renormalized perturbation theory fully in the 
\(\bJ\)-formalism and to
study how anomalies, Ward identities and gauge fixing are represented in real
Kähler language. Finally, since the real formulation makes the symplectic
structure explicit, it may be useful for comparing quantum field theory with
geometric and gravitational frameworks, where the fundamental variables are
usually real.
\\

As the investigation of the second aspect of realification of QFT models, we have studied the realification of  the Standard Model of elementary particles.
Our choice of the Standard Model for testing purely real formulations stems from its unique status. As a realistic theory of fundamental interactions, it warrants rigorous scrutiny from both phenomenological and conceptual standpoints.
\\

The internal gauge representations of all Standard Model matter fields and of the Higgs field are complex unitary representations. Whenever these fields transform nontrivially under the non-Abelian factors of the Standard Model gauge group, they belong to fundamental, conjugate-fundamental, or singlet representations of the corresponding subgroups. This structure is strongly supported by experiment. In particular, experimental data support the assignment of quarks to the fundamental color representation $\mathbf{3}$ of $SU(3)_c$, left-handed fermions to the fundamental weak representation $\mathbf{2}$ of $SU(2)_L$, and right-handed fermions to weak singlets. These assignments are inferred from measured particle multiplicities, charged- and neutral-current couplings, decay rates, jet radiation patterns, and the experimentally determined group-theoretic color factors.
 The Higgs data strongly support a dominant
$(\mathbf{1},\mathbf{2})_{1/2}$ component.
\\

After replacing each complex representation space by its underlying real vector space, the corresponding unitary representations become simultaneously orthogonal and symplectic and preserve the complex structure $\bJ$. Thus, in a real formulation, one may say that the Standard Model matter and Higgs fields carry representations that are both orthogonal and symplectic. More precisely, this ortho-symplectic structure arises through the realification of the conventional complex unitary representations.
\\

In conclusion, we have established that Standard Model  can be consistently formulated over the real numbers without any loss of physical content, providing a unified geometric framework that bridges quantum theory and classical symplectic geometry.
Whether this geometric reformulation ultimately offers new physical predictions remains unknown, but it provides a coherent mathematical language that may prove valuable, especially in contexts where real geometric structures are already fundamental, such as geometric quantization or some approaches to quantum gravity.
\\

Given constructed in this paper the realification of the Standard Model, it seems that there are the following  promising directions that go beyond simply rewriting the Standard Model. 
\begin{itemize}
\item
One could investigate whether the orthogonal-symplectic formulation
$U(N)\cong O(2N)\cup Sp(2N,R)$
reveals hidden geometric structures of gauge theories that are not as transparent in the complex formulation. For example, one could study whether renormalization, anomalies, BRST symmetry, or supersymmetry admit more natural geometric interpretations in terms of the Kähler structure $(g,\Omega,\bJ)$. Such results would constitute genuine conceptual advances rather than mere reformulations, and would considerably strengthen the motivation for the real approach.
\item
Formulating the Standard Model entirely in real terms unveils profound structural insights. It highlights the theory's inherent connections to orthogonal groups (e.g., SO(10)) and symplectic groups, as well as to more general real algebraic structures. Such a formulation is especially advantageous for model-building beyond the Standard Model—most notably in Grand Unified Theories (GUTs) and supersymmetric (SUSY) extensions, where real representations frequently provide a more transparent and unified starting point.
\item 
Any complex theory can be rewritten over the reals by doubling dimensions, resulting in a formulation with an intrinsic $\bJ$-symmetry. The real reformulation naturally allows for the breaking of this symmetry. While in the Standard Model the $\bJ$ symmetry is preserved, its breaking would signal new physics beyond it. Indeed, from the perspective of a fundamental real-field approach, there is no requirement that the number of real fields be even.  A complex theory—realified as an even-degree-of-freedom system—can be embedded into an extended theory with an odd number of real degrees of freedom.
Therefore,  $\bJ$-symmetry breaking is a natural signal of beyond the Standard Model physics.
\item 
 In the complex formulation, the mass matrix is diagonalized as
$
U_{fL}^\dagger M_f U_{fR} = M_f^{\mathrm{diag}},
$
with nonnegative real entries. The realified unitary matrices 
$\mathbb{U}_{fL}=\mathcal{R}(U_{fL})$ and $\mathbb{U}_{fR}=\mathcal{R}(U_{fR})$ 
are orthogonal and symplectic, yielding
$
\mathbb{U}_{fL}^T \mathbb{M}_f \mathbb{U}_{fR}
=
\mathcal{R}(M_f^{\mathrm{diag}}).
$
Since $M_f^{\mathrm{diag}}$ is real,
\[
\mathcal{R}(M_f^{\mathrm{diag}})
=
\begin{pmatrix}
M_f^{\mathrm{diag}} & 0\\
0 & M_f^{\mathrm{diag}}
\end{pmatrix},
\]
so each physical fermion mass appears twice. This doubling is purely a consequence of representing one complex degree of freedom by two real coordinates; it is not a physical doubling, unlike the spectral doubling in real quantum mechanics \cite{Arefeva:2025zbx}. If $\mathbb{J}$-symmetry is broken, these paired masses can split, and the corresponding fermions must be treated as independent degrees of freedom.

\item The real formalism has the potential to yield genuinely new theories, for instance, by promoting 
$\bJ$
to a dynamical field 
$\bJ(x)$ \cite{Giardino:2023uzp}, particularly in the presence of gravitational interactions.
\item Let us also note that, in the context of the $C^*$-
algebraic approach to QM and QFT, real 
$C^*$-algebras have structures that are essentially different from their complex counterparts \cite{Rosenberg:2015kia}.

Let us also note that, in the context of the 
$C^*$-algebras have structures that are essentially different from their complex counterparts \cite{Rosenberg:2015kia}.
\end{itemize}

In the end let us  list general potential advantages to formulating the quantum models
entirely over the real numbers.
\begin{itemize}
\item
{\it Closer compatibility with classical field theory}.
Classical field theory is ordinarily formulated over $\bR$.
In standard quantization, the imaginary unit
$i$
appears explicitly in commutation relations, propagators, the Schr\"odinger
equation, and the path-integral phase.

In the pure real formulation, both classical and quantum theories are written
over the real numbers. The distinction is encoded in the additional geometric
operator
$\bJ$.
Thus one may say that 
classical theory is over $\bR$ is mapped 
to quantum theory over $(\bR,\bJ)$.
This may make the passage from classical to quantum theory conceptually more
continuous. In this context it would be fruitful to consider the Hamiltonian approach \cite{VVK-OGS}.
\item {\it Numerical advantage}. Numerical calculations are ultimately performed using real floating-point
numbers, even when complex arithmetic is used at a higher level. However, no general computational improvement is guaranteed.  Numerical advantages, if any,  must be tested in concrete algorithms rather
than assumed.
\item {\it Lattice gauge theory}. For lattice gauge theory, real block-matrix representations may sometimes be
useful for implementing unitary transformations as orthogonal transformations.

\item {\it Possible extensions to other number systems}.
The usual quantum theories are formulated over
$\bC$. The real formulation instead uses
$
(\bR,\bJ),\,
\bJ^2=-\bI.
$
This separation between the scalar field and the complex-structure operator may
be useful when considering more general mathematical settings, such as
p-adic theories \cite{VVZ},
non-Archimedean models,
finite-field analogues,
or other algebraic generalizations.
Whether a physically meaningful Standard Model can be constructed in such
settings remains an open problem.
\end{itemize}
$\,$
\\

We can hope that the real formulation and its associated ortho-symplectic representations can become a usable tool in modern theoretical particle physics, especially in non-perturbative and unification-oriented researchs.

\section*{Acknowledgment}
 We are grateful to  D. Ageev and A. Trushechkin for useful discussions and remarks.
This work is supported by the Russian Science Foundation (24-11-00039, Steklov Mathematical Institute).
\newpage
\appendix
\newpage
\section{Isomorphism between $\bC$ and $\bR^2$}\label{sec:C-R2}

We define the explicit isomorphism between $\bC$ and $\bR^2$ by
\be\label{u+iv}
\cR(u+iv)=\begin{pmatrix}
  u\\v  
\end{pmatrix},\ee
with inverse $\gamma\begin{pmatrix} u\\v \end{pmatrix}=u+iv$.

Under this isomorphism, the imaginary unit $i$ is mapped to the matrix
\be\label{Ri}
\cR(i)=\bJ,\ee 
where
\be\label{bJ}
\bJ=\begin{pmatrix}
    0&-1\\1&0
\end{pmatrix},\qquad \bJ^2=-I.\ee
To preclude any confusion between \eqref{u+iv} and \eqref{Ri}, we assume that the $i$ in \eqref{Ri} is multiplied by the identity operator.

\begin{lemma}\label{lemma:1}
For $\alpha\in \bR$ and $u,v\in\bR$, we have
\be \gamma\left( e^{\bJ\alpha} \begin{pmatrix} u\\v \end{pmatrix} \right)=e^{i\alpha}(u+iv),\ee
where $\bJ$ is given in \eqref{bJ}.
\end{lemma}

The proof follows immediately from the series expansion of the exponential, or equivalently, from the standard representation of $e^{\bJ\alpha}$ as the rotation matrix
\[
e^{\bJ\alpha} =
\begin{pmatrix}
\cos\alpha & -\sin\alpha \\
\sin\alpha & \cos\alpha
\end{pmatrix}.
\]
\section{Real and Complex Fock Spaces. }\label{sec:KFock}

\subsection{Mapping $\gamma$ between Real and complex Fock spaces}\label{subsec:map}
Let $L^2_{\mathbb R}(\mathbb R^d)$ and $L^2_{\mathbb C}(\mathbb R^d)$ denote the spaces of square-integrable real- and complex-valued functions on $\mathbb R^d$, respectively.
Since any $f\in L^2_{\bC}(\bR^d )$ can be written as 
$
f=u+iv,$
where $u,v \in L^2_{\bR}(\bR^d )$, there exists a natural mapping 
\be\label{u,v}
\gamma \begin{pmatrix}
  u\\v  
\end{pmatrix}=u+iv,\ee
and
\be \gamma: L^2_{\bR}(\bR^d )\oplus  L^2_{\bR}(\bR^d )\to L^2_{\bC}(\bR^d ),
\ee
i.e.
$$L^2_{\mathbb C}(\mathbb R^d )
\simeq
 L^2_{\mathbb R}(\mathbb R ^d)\otimes_{\mathbb R}\mathbb C
 \simeq
 L^2_{\mathbb R}(\mathbb R )\oplus L^2_{\mathbb R}(\mathbb R ).$$

The complex and real bosonic Fock spaces are defined by
\[
\mathcal F_{\mathbb C}
\left(
L^2_{\mathbb C}(\mathbb R^d )
\right)
=
\bigoplus_{n=0}^{\infty}
\left(
L^2_{\mathbb C}(\mathbb R ^d)
\right)^{\otimes^n_{\mathbb C,\mathrm{Sym}}},
\qquad
\mathcal F_{\mathbb R}
\left(
L^2_{\mathbb R}(\mathbb R ^d)
\right)
=
\bigoplus_{n=0}^{\infty}
\left(
L^2_{\mathbb R}(\mathbb R^d )\right)^{\otimes^n_{\mathbb R,\mathrm{Sym}}}.
\]
Because complexification commutes with symmetric tensor products, we have
\[
\mathcal F_{\mathbb C}
\left(
L^2_{\mathbb C}(\mathbb R^d )
\right)
\simeq
\mathcal F_{\mathbb R}
\left(
L^2_{\mathbb R}(\mathbb R^d )
\right)
\otimes_{\mathbb R}\mathbb C .
\]
Hence, as real Hilbert spaces,
\be\label{isoH}
\mathcal F_{\mathbb C}
\left(
L^2_{\mathbb C}(\mathbb R ^d)
\right)
\simeq
\mathcal F_{\mathbb R}
\left(
L^2_{\mathbb R}(\mathbb R^d )
\right)
\oplus
\mathcal F_{\mathbb R}
\left(
L^2_{\mathbb R}(\mathbb R^d )
\right).
\ee

We define the doubled real Fock space
\[
\bF_{\mathbb R}
\equiv
\mathcal F_{\mathbb R}
\left(
L^2_{\mathbb R}(\mathbb R^d )
\right)
\oplus
\mathcal F_{\mathbb R}
\left(
L^2_{\mathbb R}(\mathbb R^d )
\right).
\]
To encode the complex structure, we introduce on $\bF_{\mathbb R}$ the operator
\begin{equation}
\bJ \begin{pmatrix}
  U\\V  
\end{pmatrix}
=
\begin{pmatrix}
  -V\\U 
\end{pmatrix},
\qquad
\bJ^2=-I,
\label{bJ_main}
\end{equation}
which is represented by the $2\times 2$ matrix
\[
\bJ=
\begin{pmatrix}
  0&-1\\1&0 
\end{pmatrix}.
\]
This corresponds precisely to multiplication by $i$ under the identification $U+iV$.

Equipped with $\bJ$, the space $\bF_{\mathbb R}$ becomes a Kähler vector space. For
\be\bPsi^{i}=
\begin{pmatrix}
  U^{i}\\V^{i}  
\end{pmatrix},\quad i=1,2,
\ee
we define the real scalar product
\be
g(\bPsi^{1},\bPsi^{2})
=
(U^{1},U^{2})_{\mathcal F_{\mathbb R}}
+
(V^{1},V^{2})_{\mathcal F_{\mathbb R}},
\ee
and the symplectic form
\be
\omega(\bPsi^{1},\bPsi^{2})
=
(V^{1},U^{2})_{\mathcal F_{\mathbb R}}
-
(U^{1},V^{2})_{\mathcal F_{\mathbb R}},
\ee
where $(\cdot,\cdot)_{\mathcal F_{\bR}}$ denotes the inner product on $\mathcal F_{\mathbb R}$.
We denote this Kähler Fock space by
\[
\mathcal K\mathcal F_{\mathbb R}
\equiv
\left(
\bF_{\mathbb R}, \bJ
\right).
\]
Consequently, we have the complex Hilbert space isomorphism
\begin{equation}
\mathcal F_{\mathbb C}
\left(
L^2_{\mathbb C}(\mathbb R ^d)
\right)
\simeq
\mathcal K\mathcal F_{\mathbb R}
\left(
L^2_{\mathbb R}(\mathbb R ^d)
\right).
\label{main_iso}
\end{equation}

The explicit realization of the isomorphism \eqref{main_iso}, including the component-wise decomposition $\Psi = U + i V$, is relegated to next subsection~\ref{app:Fock_details} for completeness.

\subsection{Explicit Form of the Real and Complex Fock Space \\Isomorphism}\label{app:Fock_details}
 Here we spell out the explicit maps realizing the isomorphism \eqref{main_iso}. 
Defining \be\label{B1}\cR(\Psi) = \bPsi =\begin{pmatrix}
  U\\V  
\end{pmatrix} \quad \mbox{and} \quad \gamma(\bPsi) = U+iV,
\ee 
we can realized the isomorphism \eqref{isoH}   by the map $\gamma$
\be
\cR:\mathcal F_{\mathbb C}
\left(
L^2_{\mathbb C}(\mathbb R ^d)
\right)
\to
\mathcal F_{\mathbb R}
\left(
L^2_{\mathbb R}(\mathbb R ^d )
\right)
\oplus
\mathcal F_{\mathbb R}
\left(
L^2_{\mathbb R}(\mathbb R ^d )
\right),
\ee
and
\be
\gamma:
\mathcal F_{\mathbb R}
\left(
L^2_{\mathbb R}(\mathbb R ^d )
\right)
\oplus
\mathcal F_{\mathbb R}
\left(
L^2_{\mathbb R}(\mathbb R ^d )
\right)
\to
\mathcal F_{\mathbb C}
\left(
L^2_{\mathbb C}(\mathbb R ^d)
\right).
\ee
These maps are inverse to each other, and $\gamma$ intertwines multiplication by $i$ with the action of $\bJ$ defined in \eqref{bJ_main}. This establishes the complex Hilbert space isomorphism
\[
\mathcal F_{\mathbb C}
\left(
L^2_{\mathbb C}(\mathbb R ^d)
\right)
\simeq
\mathcal K\mathcal F_{\mathbb R}
\left(
L^2_{\mathbb R}(\mathbb R ^d)
\right).
\]
Explicitly, let
$
\Psi\in
\mathcal F_{\mathbb C}
\left(
L^2_{\mathbb C}(\mathbb R ^d)
\right).
$
Then
$
\Psi=(\Psi_0,\Psi_1,\Psi_2,\ldots)$,
where $
\Psi_n\in
L^2_{\bC,\mathrm{sym}}
\left(
\mathbb R ^n
\right), \, n=0,1,2\ldots$.
Each \(n\)-particle wave function can be written as
\[
\Psi_n(\bk_1,\ldots,\bk_n)
=
U_n(\bk_1,\ldots,\bk_n)
+
iV_n(\bk_1,\ldots,\bk_n),
\]
where
\[
U_n,V_n\in
L^2_{\mathrm{sym}}
\left(
\mathbb R ^n;\mathbb R
\right).
\]
Thus
\[
\Psi=U+iV,
\]
where
\[
U=(U_0,U_1,U_2,\ldots),
\qquad
V=(V_0,V_1,V_2,\ldots),
\]
and
\[
U,V\in
\mathcal F_{\mathbb R}
\left(
L^2_{\mathbb R}(\mathbb R ^d)
\right).
\]
Therefore the isomorphism is given  by $\cR$ 
with the inverse given by $\gamma$ as defined in \eqref{B1}.

\newpage
\section{$\bJ$-Fourier Transform }\label{app:JFT}
We define the $\bJ$-Fourier transform by formula 
\be\label{KFc}
\begin{pmatrix}
 \tilde{f}^{\bJ}_1(p)  \\ \tilde{f}^{\bJ}_2(p)
\end{pmatrix}= \int dx\, e^{-\bJ p x}\,\begin{pmatrix}
 f_1(x)  \\ f_2(x)
\end{pmatrix} 
\ee
where $\bJ = \begin{pmatrix}0&-1\\1&0\end{pmatrix}$ satisfies $\bJ^2 = -\id$; or in short   
\bea\label{KF}
\tilde{f}^{\bJ}(p) = \int dx\, e^{-\bJ p x}\,f(x), \\ f(x): \bR\to \bR^2, \quad \tilde{f}^{\bJ}(p): \bR\to  \bR^2
\eea
\\

The inverse Fourier transform is
\be\label{IKF}
f(x) = \frac{1}{2\pi}\int dp\, e^{\bJ p x}\,\tilde{f}^{\bJ}(p)\, \quad f,\tilde{f}^{\bJ} \in \bR^2
\ee

To prove the formula \eqref{IKF} we have to prove that
\be\label{D4}
\frac1{2\pi}\int dp e^{\bJ p (x_1-x_2)}=\id \,\delta (x_1-x_2)\ee

We want to prove
\[
\frac{1}{2\pi}\int_{-\infty}^{\infty} dp\; e^{\bJ p a} = \id\,\delta(a),
\]
where \(a = x_1-x_2\), \(\bJ = \begin{pmatrix}0&-1\\1&0\end{pmatrix}\) satisfies \\\(\bJ^2 = -\id\), and the integral is understood in the sense of distributions.

Using the identity \(e^{\bJ \theta} = \id\cos\theta + \bJ\sin\theta\) (which follows from the power series expansion because \(\bJ^{2n}=(-1)^n\id\) and \(\bJ^{2n+1}=(-1)^n\bJ\)), we have
\[
e^{\bJ p a} = \id \cos(p a) + \bJ \sin(p a).
\]

Thus,
\bea
&&\frac{1}{2\pi}\int_{-\infty}^{\infty} dp\; e^{\bJ p a}\\\nn
&&= \id \cdot \frac{1}{2\pi}\int_{-\infty}^{\infty} dp\;\cos(p a) \;+\; \bJ \cdot \frac{1}{2\pi}\int_{-\infty}^{\infty} dp\;\sin(p a).
\eea

The integral of \(\sin(p a)\) over a symmetric interval vanishes because it is an odd function:
\[
\int_{-\infty}^{\infty} \sin(p a)\,dp = 0.
\]

For the cosine integral, we recall the Fourier representation of the Dirac delta:
\[
\int_{-\infty}^{\infty} e^{i p a}\,dp = 2\pi\,\delta(a).
\]
Since \(e^{i p a} = \cos(p a) + i\sin(p a)\) and the sine integral is zero, we obtain
\[
\int_{-\infty}^{\infty} \cos(p a)\,dp = 2\pi\,\delta(a).
\]

Therefore,
\[
\frac{1}{2\pi}\int_{-\infty}^{\infty} \cos(p a)\,dp = \delta(a).
\]

Substituting back,
\[
\frac{1}{2\pi}\int_{-\infty}^{\infty} dp\; e^{\bJ p a} = \id\,\delta(a) + \bJ \cdot 0 = \id\,\delta(a).
\]

Thus the desired formula holds:
\be
\frac{1}{2\pi}\int_{-\infty}^{\infty} dp\; e^{\bJ \,p (x_1-x_2)} = \id\,\delta(x_1-x_2).
\ee

\subsection {$\bJ$-Fourier Transform of a Schwartz Distribution }\label{app:JFT1}
\begin{definition}[{\it Schwartz space}]
\it Let \(\mathcal{S}(\mathbb{R}^n)\) denote the space of all \(C^\infty\) functions \(\phi:\mathbb{R}^n\to\mathbb{C}\) such that for all multi-indices \(\alpha,\beta\in\mathbb{N}_0^n\)
\[
\sup_{x\in\mathbb{R}^n} |x^\alpha \partial^\beta \phi(x)| < \infty.
\]
\end{definition}
It is known that \cite{GS,VVS} the Fourier transform is an isomorphism of \(\mathcal{S}(\mathbb{R}^4)\).

\begin{theorem}\label{th:1}

Let $\bF=\begin{pmatrix}
    f_1\\f_2
\end{pmatrix}$ and each component
$f_i\in\cS(\bR^4)$, i=1,2. Then each component of the 2-component vector 
\be\label{Jft}
\tilde \bF^\bJ(k^0,\mathbf{k}) = \int_{R^4} e^{-\bJ(k^0 t - \mathbf{k}\cdot\mathbf{x})}\, \bF(t,\bx)\, d^4x
\ee
belongs to \(\mathcal{S}(\mathbb{R}^4)\).
\end{theorem}

Hence, the $\bJ$-Fourier transform is an isomorphism of $\cS(\bR^4)\oplus \cS(\bR^4)$ onto itself. This theorem is the analog of the theorem according which the usual Fourier transform is an isomorphism of \(\mathcal{S}(\mathbb{R}^4)\).
\subsection{The restriction of the $\bJ$-Fourier Transform to the Mass Shell}\label{app:JFT2}
It is known, see \cite{Hormander}, the following theorem
\begin{theorem}\label{th:2}
 If $f\in\cS(\bR^4)$ then the function
\be
\tilde f(\mathbf{k}) = \int_{\mathbb{R}^4} f(t,\mathbf{x})\, e^{-i(\omega_{\mathbf{k}} t - \mathbf{k}\cdot\mathbf{x})}\, d^4x
\ee
belongs to \(\mathcal{S}(\mathbb{R}^3)\).
\end{theorem}
Below we formulate and prove the $\bJ$-analogue of this theorem.

\begin{theorem}\label{th:3}
 Let $\bF=\begin{pmatrix}
    f_1\\f_2
\end{pmatrix}$ and each component
$f_i\in\cS(\bR^4)$, i=1,2 and let \(m\ge 0\) be fixed. Define \(\omega_{\mathbf{k}} = \sqrt{m^2 + |\mathbf{k}|^2}\) for \(\mathbf{k}\in\mathbb{R}^3\).  Then each component of the 2-component vector
\be
\tilde\bF^\bJ(\mathbf{k}) = \int_{\mathbb{R}^4}  e^{-\bJ(\omega_{\mathbf{k}} t - \mathbf{k}\cdot\mathbf{x})}\,\bF(t,\mathbf{x})\, d^4x
\ee
belongs to \(\mathcal{S}(\mathbb{R}^3)\).
\end{theorem}

\begin{proof}
According the Theorem\ref{th:1} $\tilde \bF^\bJ_i(k^0,\mathbf{k})\in\mathcal{S}(\mathbb{R}^4)$, $i=1,2$. Observe that
\be
\bF^\bJ_i(\bk) = \tilde \bF^\bJ_i(\omega_{\mathbf{k}},\mathbf{k}),\quad i=1,2,
\ee
i.e. $\bG^\bJ_i$ are the restrictions of the Schwartz function $\tilde \bF^\bJ_i$ to the mass shell \(\Sigma = \{(k^0,\mathbf{k})\in\mathbb{R}^4: k^0 = \omega_{\mathbf{k}}\}\).

We must show that the maps $\bk\mapsto \tilde \bF_i^\bJ(\omega_{\bk},\bk)$  belong to \(\mathcal{S}(\mathbb{R}^3)\). This is a consequence of the fact that the mass shell is a smooth (for \(m>0\)) or almost smooth (for \(m=0\)) submanifold and that the restriction of a Schwartz function to such a submanifold, when parametrized by \(\mathbb{R}^3\), yields a Schwartz function on \(\mathbb{R}^3\).
The prof is the same as in the standard Fourier transform.
\end{proof}


\newpage

\subsection{$\bJ$-Sokhotski–Plemelj Formula}

We prove the matrix version of the Sokhotski formula - the $\bJ$-Sokhotski formula:
\be\label{KSp}
\lim_{\varepsilon \to 0^+} \frac{1}{x \pm \bJ\varepsilon} = \mathcal{P}\frac{1}{x} \mp \bJ\pi\,\delta(x),
\ee
where \(x\) is real, \(\bJ = \begin{pmatrix}0&-1\\1&0\end{pmatrix}\) satisfies \(\bJ^2 = -\id\), and the limit is understood in the sense of distributions. The proof follows the same steps as the complex case, because \(\bJ\) plays the role of \(i\).
\\

For \(\varepsilon>0\), we write
\be\label{C20}
\frac{1}{x + \bJ\varepsilon} = \frac{x - \bJ\varepsilon}{x^2 + \varepsilon^2}.
\ee
Let \(\varphi(x)\) be a smooth, compactly supported test function. Consider the integral
\be\label{C21}
I(\varepsilon) = \int_{-\infty}^{\infty} \frac{\varphi(x)}{x + \bJ\varepsilon}\,dx = \int_{-\infty}^{\infty} \frac{x\varphi(x)}{x^2+\varepsilon^2}\,dx \;-\; \bJ\varepsilon \int_{-\infty}^{\infty} \frac{\varphi(x)}{x^2+\varepsilon^2}\,dx.
\ee
Introducing  notation 
\be
J_\varepsilon=\int_{-\infty}^{\infty} \frac{\varepsilon \varphi(x)}{x^2+\varepsilon^2}\,dx\ee
one can see \cite{VVS} that  \be
\label{C23}\lim_{\varepsilon \to 0^+} J_\varepsilon =  \pi \varphi(0).
\ee
Hence the second term in \eqref{C21}
is
$ -\bJ\pi \varphi(0)$.
The limit of the first term then $\epsilon\to 0$  is
$\Pv \int_{-\infty}^{\infty} \frac{\varphi(x)}{x} \, dx$. Together with \eqref{C23} this gives \eqref{KSp} for $+$ in the LHS.
\\

Replacing \(\varepsilon\) by \(-\varepsilon\) gives
\[
\frac{1}{x - \bJ\varepsilon} = \frac{x + \bJ\varepsilon}{x^2+\varepsilon^2},
\]
and the same calculation yields the plus sign in front of the \(\bJ\pi\delta(x)\) term:
\[
\lim_{\varepsilon\to0^+} \frac{1}{x - \bJ\varepsilon} = \mathcal{P}\frac{1}{x} + \bJ\pi\,\delta(x).
\]
Thus the matrix $\bJ$-Sokhotski formula is proved.
\\

We also have
distributional identities
\be
\lim_{\varepsilon\to0^+}\frac{1}{-x^2+\bJ\,\varepsilon \,x^0}
=
-\cP\frac{1}{x^2}
-
\pi\,\operatorname{sgn}(x^0)\,
\delta(x^2)\,\bJ ,
\ee
\be\label{4.45}
 \lim _{\varepsilon\to 0}\frac{x }{x^2+\epsilon ^2} =\cP(\frac{1}{x});\qquad \lim _{\varepsilon\to 0}\frac{\epsilon }{x^2+\epsilon ^2}=\pi\delta(x),
   \ee
\newpage

\newpage

\section{$\bJ$-Cutkosky Rules }

It is well known that imaginary part of a loop amplitude arises when the internal particles in the loop go "on-shell," meaning they become real, physical particles. The Cutkosky rules formalize this by instructing us to "cut" through the diagram in all possible ways such that the cut lines can be put on-shell.
\\

The $\bJ$-Cutkosky rules give the discontinuity  $\bJ$-part of a Feynman
$\bJ$-diagram across a physical branch cut. They are in fact the diagrammatic form of
the orthogonality  of the $\bS$-matrix.
\\

The $\bJ$-Cutkosky
rules are the Feynman-diagram version of the orhogonality  formula  formula \eqref{UnT} that is
\be\label{UnTr}
 \bJ(\bT^T-\bT)=\bT^T \bT\ee
For a scalar propagator
\bea \label{bGr}
\tilde\bG^{\bJ}_{F}(p)=
 \bJ\,\cP(\frac{1}{p^2 - m^2 })+\pi \delta(p^2 - m^2 )
\eea
the cutting rule is
\be
\frac{\bJ}{k^2-m^2+\bJ 0}
\quad\longrightarrow\quad
2\pi\,\delta_+(k^2-m^2)\equiv 2\pi\,
\theta(k^0)\,\delta(k^2-m^2).
\ee
 here $\theta(k^0)$ is added as in the usual case to remove the negative-energy term.

The general Cutkosky formula has the schematic form
\be\label{disc}
\operatorname{Disc}\bM
=
\sum_{\text{cuts}}
\int d\Phi_{\text{cut}}\,
\bM_L\,\bM_R^T
\ee
where
\[
\operatorname{Disc}\bM
=
\bM(s+\bJ\,0)-\bM(s-\bJ\,0)
=
2\bJ\,(\bM^T-\bM)
\]
for a real physical variable \(s\). The two factors
$\bM_L$ and $\bM_R^*$ are the amplitudes on the two sides
of the cut diagram.
Here 
the symbol
$
d\Phi_{\text{cut}}
$
denotes the Lorentz-invariant phase-space measure of the particles crossing the
cut.
\newpage

\section{Realification of the Groups of the Standard Model}\label{app:GSM}
\subsection{Realification of $U(1)$}\label{app:E1}
\subsubsection*{General Realification Formula}
An element of $U(1)$ is
\begin{equation}
U(\theta)=e^{i\theta},
\qquad
\theta\in\mathbb{R},
\end{equation}
acting on
\begin{equation}
z=x+iy\in\mathbb{C}
\end{equation}
by
\begin{equation}
z\mapsto e^{i\theta}z.
\end{equation}

Since
\begin{equation}
e^{i\theta}=\cos\theta+i\sin\theta,
\end{equation}
we obtain
\begin{equation}
e^{i\theta}(x+iy)
=
(\cos\theta\,x-\sin\theta\,y)
+i(\sin\theta\,x+\cos\theta\,y).
\end{equation}

Identifying
\begin{equation}
\mathbb{C}\cong\mathbb{R}^2,
\qquad
x+iy\leftrightarrow
\begin{pmatrix}x\\y\end{pmatrix},
\end{equation}
the transformation is represented by
\begin{equation}
\cR(e^{i\theta})=
\begin{pmatrix}
\cos\theta&-\sin\theta\\
\sin\theta&\cos\theta
\end{pmatrix}.
\end{equation}
Thus $\cR:\,U(1)\longrightarrow GL(2,\mathbb{R}).$

For any complex number
$
u=a+ib
$,
\be
\cR(a+ib)=
\begin{pmatrix}
a&-b\\
b&a
\end{pmatrix}.
\ee
For $u=e^{i\theta}$,
$
a=\cos\theta$,
$b=\sin\theta,$
hence
\begin{equation}
\cR(e^{i\theta})=
\begin{pmatrix}
\cos\theta&-\sin\theta\\
\sin\theta&\cos\theta
\end{pmatrix}.
\end{equation}
Moreover,
$
\cR(u_1u_2)=\cR(u_1)\cR(u_2)$,
so $\cR$ is a faithful representation.

\subsubsection*{Orthogonal Property}
The realified matrix satisfies
\begin{equation}
\cR(e^{i\theta})^T\cR(e^{i\theta})=I_2,
\end{equation}
and
\begin{equation}
\det\cR(e^{i\theta})
=\cos^2\theta+\sin^2\theta=1.
\end{equation}
Therefore
\be\label{cRU1}
\cR(U(1))=SO(2),
\ee
so that
$
U(1)\cong SO(2)$.
\subsubsection*{Symplectic Property}
Since $
\cR(e^{i\theta})=e^{\theta \bJ}=I_2\cos\theta+\bJ\sin\theta$
we have
\be
\cR(e^{i\theta})^T
J
\cR(e^{i\theta})=J,
\end{equation}
and hence
$
\cR(U(1))\subset Sp(2,\mathbb{R}).
$
Together with \eqref{cRU1} we get 
\be
\cR(U(1))
=
SO(2)\cap Sp(2,\mathbb{R}).
\ee

\subsubsection*{Lie Algebra}
The Lie algebra is
\begin{equation}
\mathfrak{u}(1)=i\mathbb{R}.
\end{equation}
A generic element is
\begin{equation}
X=i\alpha,
\qquad
\alpha\in\mathbb{R},
\end{equation}
whose realification is
\begin{equation}
\cR(i\alpha)
=
\alpha
\begin{pmatrix}
0&-1\\
1&0
\end{pmatrix}
=\alpha \bJ.
\end{equation}
Thus
\begin{equation}
\mathfrak{u}(1)\cong\mathfrak{so}(2)=\mathbb{R}J.
\end{equation}

\subsubsection{Application: Charged Scalar Field and Real Covariant Derivative}
Let
$
\phi=\phi_1+i\phi_2$.
Introduce the real doublet
\begin{equation}
\Phi=
\begin{pmatrix}
\phi_1\\
\phi_2
\end{pmatrix}.
\end{equation}
Under a local $U(1)$ transformation,
$
\phi\mapsto e^{iq\alpha(x)}\phi,
$
the real field transforms as
$
\Phi\mapsto e^{q\alpha(x)\cJ}\Phi,
$
or explicitly,
\begin{equation}
\Phi\mapsto
\begin{pmatrix}
\cos(q\alpha)&-\sin(q\alpha)\\
\sin(q\alpha)&\cos(q\alpha)
\end{pmatrix}\Phi.
\end{equation}

The complex covariant derivative
$D_\mu=\partial_\mu+iqA_\mu$
becomes
$
\mathcal D_\mu
=
\partial_\mu I_2
+
qA_\mu \bJ.
$
It acts as
\be
\mathcal D_\mu\Phi=
\begin{pmatrix}
\partial_\mu\phi_1-qA_\mu\phi_2\\
\partial_\mu\phi_2+qA_\mu\phi_1
\end{pmatrix}.
\ee

\subsubsection*{Summary}
The realification of $U(1)$ is
\begin{equation}
e^{i\theta}
\longmapsto e^{\theta \bJ}=
\begin{pmatrix}
\cos\theta&-\sin\theta\\
\sin\theta&\cos\theta
\end{pmatrix},
\end{equation}
i.e.
\be
\cR(U(1))= SO(2)=SO(2)\cap Sp(2,\mathbb{R}),
\ee
and
\be
\mathfrak{u}(1)\cong\mathfrak{so}(2).
\ee
\subsection{Realification of $SU(2)$}\label{app:E2}
\subsubsection{ Explicit realification of $SU(2)$}
We write explicitly the realification of \(SU(2)\). We take its standard complex fundamental representation on \(\mathbb C^2\) and regard \(\mathbb C^2\) as a real vector space \(\mathbb R^4\) by forgetting the complex structure. Every \(U \in SU(2)\) then becomes a real \(4\times 4\) orthogonal matrix that also preserves the complex structure \(J\). The resulting subgroup of \(GL(4,\mathbb R)\) is isomorphic to \(SU(2)\) and is given explicitly below.

To this purpose we write parametrization of $SU(2)$.
Any $U \in SU(2)$ can be written as
\be\label{U}
U = \begin{pmatrix}
a & b \\
-\bar{b} & \bar{a}
\end{pmatrix},
\qquad a,b \in \mathbb C, \quad |a|^2 + |b|^2 = 1.
\ee
Writing
$
a = a_0 + i a_1,\quad b = b_0 + i b_1,
$
with $a_0,a_1,b_0,b_1 \in \mathbb R$. 
The matrix $U$ in \eqref{U} has  the form 
\be\label{MM}
M=A+iB,
\qquad
A,B\in M_2(\bR),
\ee
where 
\be
A = 
\begin{pmatrix}
a_0  &  b_0 \\
- b_0 & a_0 
\end{pmatrix},\quad
B=
\begin{pmatrix}
  a_1 &   b_1\\
  b_1 &  -  a_1
\end{pmatrix}.
\ee
According a general consideration a natural realification is defined by (see \cite{Ying:2025xyl}) 
\be\label{RMM}
\cR(M)
=\begin{pmatrix}
A&-B\\
B&A
\end{pmatrix}
\ee
Therefore,
\be
\cR(U)=
\begin{pmatrix}
a_0 & b_0 & -a_1 & -b_1 \\
-b_0 & a_0 & -b_1 & a_1 \\
a_1 & b_1 & a_0 & b_0 \\
b_1 & -a_1 & -b_0 & a_0
\end{pmatrix}.
\ee

The determinant of this matrix is [see Mathematica: MatrixM4.nb]
\be
\det[\cR(U)]=\left(a_0^2+a_1^2+b_0^2+b_1^2\right){}^2\ee
One can check explicitly  [see Mathematica: MatrixM4.nb] that this matrix is:
\begin{itemize}
    \item Orthogonal: \(\cR(U)^T \cR(U) = I_4\), hence \(\cR(U) \in O(4)\).
    \item Commutes with  $\bJ$, hence preserves the symplectic form: 
    \(\cR(U)^T J \cR(U) = J\), so \(\cR(U) \in Sp(4,\bR)\).
\end{itemize}
Therefore,
$
\cR(U) \in O(4) \cap Sp(4,\mathbb R) \cong U(2).
$

Moreover, the complex determinant of \(U\) is \(1\) (by definition of \(SU(2)\)). In the real \(4\times 4\) representation, this corresponds to the condition
$
\det_{\mathbb C} U = 1,
$
which picks out the subgroup \(SU(2)\) inside \(U(2)\). Thus, the set of all matrices \(\cR(U)\) for \(|a|^2+|b|^2=1\) is precisely the realification of \(SU(2)\).
\subsubsection{Application: scalar field doublet}
In the conventional complex formulation
\be\label{ReaU}
\phi(x)=
\begin{pmatrix}
\phi_1(x)\\
\phi_2(x)
\end{pmatrix}
\in\bC^2
\end{equation}
is a complex scalar field in the fundamental representation of $SU(2)$. 
Under
a local gauge transformation,
\begin{equation}
\phi(x)\longmapsto U(x)\phi(x),
\qquad
U(x)\in SU(2).
\end{equation}
Realification of the scalar doublet
\be\label{Phi12}
\phi_1
=
\frac{1}{\sqrt2}
\left(q_1+iq_2\right),
\qquad
\phi_2
=
\frac{1}{\sqrt2}
\left(q_3+iq_4\right).
\end{equation} 
is realized as 
\be\label{PhiU2}
\Phi
=
\begin{pmatrix}
q_1\\
q_2\\
q_3\\
q_4
\end{pmatrix}
\in\bR^4.
\end{equation}
Equivalently, if
\begin{equation}
\phi=u+iv,
\qquad
u,v\in\bR^2,
\end{equation}
then
\begin{equation}
\Phi=
\begin{pmatrix}
u\\v
\end{pmatrix}.
\end{equation}
Real covariant derivative 
are written using the real generators
$\bT_a=\mathcal R(t_a)$
\be
[\bT_a,\bT_b]
=
\epsilon_{abc}\bT_c,
\ee
\be
\bT_a^T=-\bT_a,
\ee
and
\be
[\bT_a,\bJ]=0.
\ee
Therefore,
\begin{equation}
\bT_a\in
\mathfrak{so}(4)\cap\mathfrak{sp}(4,\bR)
\simeq\mathfrak u(2).
\end{equation}

The real covariant derivative is

\begin{equation}
\mathcal D_\mu\Phi
=
\partial_\mu\Phi
+
gA_\mu^a\bT_a\Phi.
\end{equation}

\newpage
\subsection{Realification of $SU(3)$}\label{app:E3}
We start from the complex fundamental representation of $SU(3)$.
An element of the special unitary group $SU(3)$ is a complex
$3\times3$ matrix
\be\label{U(3)}
U=A+iB,\quad A,B\in M_3(\mathbb{R}),
\ee
satisfying
\begin{equation}
U^\dagger U=I_3,
\qquad
\det U=1.
\end{equation}
Writing
$
U=A+iB,
$
the unitarity condition is equivalent to
\be
A^TA+B^TB = I_3,
\quad
A^TB = B^TA.
\ee

The realification  is the homomorphism
$
\cR:\,SU(3)\longrightarrow GL(6,\mathbb R),
$
defined by
\be\label{cRu(3)}
\cR(U)=
\begin{pmatrix}
A&-B\\
B&A
\end{pmatrix}.
\ee
This matrix acts on the real vector space
$
\mathbb C^3\simeq\mathbb R^6.
$
Indeed, writing
\begin{equation}
z=x+iy,
\qquad
x,y\in\mathbb R^3,
\end{equation}
one has
\begin{equation}
\begin{pmatrix}
x\\
y
\end{pmatrix}
\longmapsto
\begin{pmatrix}
A&-B\\
B&A
\end{pmatrix}
\begin{pmatrix}
x\\
y
\end{pmatrix}.
\end{equation}
One can check that $\cR(U)$ is 
\begin{itemize}
\item orthogonal. Indeed, since $U$ is unitary,
$U^\dagger U=I$,
one check that 
\begin{equation}
\cR(U)^T\cR(U)
=
\begin{pmatrix}
A^TA+B^TB
&
-A^TB+B^TA
\\
-B^TA+A^TB
&
A^TA+B^TB
\end{pmatrix}
=
I_6.
\end{equation}
Therefore
$\cR(SU(3))\subset SO(6)$.
Moreover,
$\det \cR(U)
=
|\det U|^2
=
1.
$
\item
symplectic. A direct computation gives
\begin{equation}
\cR(U)^T
\bJ
\cR(U)
=
\bJ.
\end{equation}
Hence
$\cR(SU(3))
\subset
Sp(6,\mathbb R)$.
\end{itemize}

Thus every element of $\cR(SU(3))$ is simultaneously
 orthogonal and 
 symplectic
\be
\cR(SU(3))
\subset
SO(6)\cap Sp(6,\mathbb R).
\ee
One also has
\begin{equation}
\cR(SU(n))
=
\left\{M\in
SO(2n)\cap Sp(2n,\mathbb R)
:
\det_{\mathbb C}M=1\right
\}.
\end{equation}

\subsubsection{Explicit Realification of $SU(3)$ }

Let $U$ element of the fundamental representation of $SU(3)$. It can be represented as 
\begin{equation}
U(\boldsymbol{\theta})
=
\exp\left(
\frac{i}{2}H(\boldsymbol{\theta})
\right),
\end{equation}
where
\begin{equation}
H(\boldsymbol{\theta})
=\sum _{i=1}^8 \theta _i \lambda _i=
\begin{pmatrix}
\theta_3+\dfrac{\theta_8}{\sqrt{3}}
&
\theta_1-i\theta_2
&
\theta_4-i\theta_5
\\[2mm]
\theta_1+i\theta_2
&
-\theta_3+\dfrac{\theta_8}{\sqrt{3}}
&
\theta_6-i\theta_7
\\[2mm]
\theta_4+i\theta_5
&
\theta_6+i\theta_7
&
-\dfrac{2\theta_8}{\sqrt{3}}
\end{pmatrix},
\end{equation}
here $\lambda_i, i=1,...8$ are the Gell-Mann matrices.
The realification map is defined by
\begin{equation}
\cR(A+iB)
=
\begin{pmatrix}
A&-B\\
B&A
\end{pmatrix},
A,B\in M_3(\mathbb{R}).
\end{equation}

Since $\cR$ is an algebra homomorphism,
\begin{equation}
\cR(e^X)=e^{\cR(X)}.
\end{equation}
Therefore,
\begin{equation}
\cR\bigl(U(\boldsymbol{\theta})\bigr)
=
\exp\left[
\cR\left(
\frac{i}{2}H(\boldsymbol{\theta})
\right)
\right].
\end{equation}

Introduce
\begin{equation}
a=
\theta_3+\frac{\theta_8}{\sqrt{3}},
\qquad
b=
-\theta_3+\frac{\theta_8}{\sqrt{3}},
\qquad
c=
-\frac{2\theta_8}{\sqrt{3}}.
\end{equation}
and 
write
\begin{equation}
H=P+iQ,
\end{equation}
where
\begin{equation}
P=
\begin{pmatrix}
a&\theta_1&\theta_4\\
\theta_1&b&\theta_6\\
\theta_4&\theta_6&c
\end{pmatrix}
\end{equation}
is real symmetric, and
\begin{equation}
Q=
\begin{pmatrix}
0&-\theta_2&-\theta_5\\
\theta_2&0&-\theta_7\\
\theta_5&\theta_7&0
\end{pmatrix}
\end{equation}
is real antisymmetric.
Since
\begin{equation}
\frac{i}{2}H
=
-\frac{1}{2}Q+\frac{i}{2}P,
\end{equation}
its realification is
\begin{equation}
\cR\left(
\frac{i}{2}H
\right)
=
\frac{1}{2}
\begin{pmatrix}
-Q&-P\\
P&-Q
\end{pmatrix}.
\end{equation}

Consequently,
\be\label{rhoU}
\cR\bigl(U(\boldsymbol{\theta})\bigr)
=
\exp\left\{
\frac{1}{2}
\begin{pmatrix}
0&\theta_2&\theta_5
&
-a&-\theta_1&-\theta_4
\\
-\theta_2&0&\theta_7
&
-\theta_1&-b&-\theta_6
\\
-\theta_5&-\theta_7&0
&
-\theta_4&-\theta_6&-c
\\
a&\theta_1&\theta_4
&
0&\theta_2&\theta_5
\\
\theta_1&b&\theta_6
&
-\theta_2&0&\theta_7
\\
\theta_4&\theta_6&c
&
-\theta_5&-\theta_7&0
\end{pmatrix}
\right\}.
\ee
The matrix in the exponent is real and antisymmetric. Therefore,
\begin{equation}
\cR(U)^T\cR(U)=I_6.
\end{equation}
Introduce the real complex-structure matrix
\begin{equation}
\bJ=
\begin{pmatrix}
0&-I_3\\
I_3&0
\end{pmatrix}.
\end{equation}
The realified matrix commutes with $\bJ$,
\begin{equation}
\cR(U)\bJ=\bJ\cR(U),
\end{equation}
(see calculations Matrix6.nb) and consequently
\begin{equation}
\cR(U)^T J\cR(U)=J.
\end{equation}

Thus,
\begin{equation}
\cR(U)\subset SO(6)\cap Sp(6,\mathbb{R}).
\end{equation}

It is known that,
\begin{equation}
SO(6)\cap Sp(6,\mathbb{R})
\cong U(3),
\end{equation}
whereas the realified $SU(3)$ subgroup is characterized by the additional
complex determinant condition
$
\det_{\mathbb{C}}U=1.
$

Note that, consequently, in the realified Standard Model, the gauge groups $SU(2)$ and $SU(3)$ are realized as distinguished subgroups of $SO(4)\cap Sp(4,R)$ and $SO(6)\cap Sp(6,R)$, respectively, selected by the extra condition $\det
_\bC
=1$.
\\

Let us calculate the determinant of $\cR(U(\boldsymbol{\theta}))$ given by \eqref{rhoU}.
Using that the determinant of the exponential of a square matrix is given by the exponent of its trace,
\begin{equation}
\det(e^{M})=e^{\operatorname{Tr}M}.
\end{equation}
and that 
in the present case,
\begin{equation}
M=
\begin{pmatrix}
 0 & \theta _2 & \theta _5 &
   -a & -\theta _1 & -\theta _4 \\
 -\theta _2 & 0 & \theta _7 &
   -\theta _1 & -b & -\theta _6 \\
 -\theta _5 & -\theta _7 & 0 &
   -\theta _4 & -\theta _6 &
   -c \\
 a & \theta _1 & \theta _4 &
   0 & \theta _2 & \theta _5 \\
 \theta _1 & b & \theta _6 &
   -\theta _2 & 0 & \theta _7 \\
 \theta _4 & \theta _6 & c &
   -\theta _5 & -\theta _7 & 0
\end{pmatrix}.
\end{equation}
with
all diagonal entries vanish,
\begin{equation}
\operatorname{Tr}M=0,
\end{equation}
we get
\begin{equation}
\det(e^{M})
=
e^{\operatorname{Tr}M}
=
e^{0}
=
1.
\end{equation}
Moreover, the matrix $M$ is real and antisymmetric,
\begin{equation}
M^{T}=-M.
\end{equation}
Hence
\begin{equation}
(e^{M})^{T}
=
e^{M^{T}}
=
e^{-M},
\end{equation}
and therefore
\begin{equation}
(e^{M})^{T}e^{M}
=
e^{-M}e^{M}
=
I_{6}.
\end{equation}
Thus,
\begin{equation}
e^{M}\in O(6).
\end{equation}
Since its determinant is equal to $1$, it follows that
\begin{equation}
e^{M}\in SO(6).
\end{equation}
Therefore, the determinant of the $6\times6$ realified matrix is always
\begin{equation}
\det\!\left(\exp(M)\right)=1,
\end{equation}
independently of the values of the parameters
$a$, $b$, $c$, and
$\theta_1,\ldots,\theta_7$.

We  have also checked [Mathematica: Matrix6.nb] that 
\be
M\bJ=\bJ M.\ee

\newpage

\section{Realification of  Fermions Fields}\label{app:fermi}
\subsection{Realification of Dirac Fermions}
The realification of the Dirac field is an  reformulation of the standard complex Dirac theory entirely over the real numbers. Instead of a four-component complex spinor, one introduces an equivalent eight-component real spinor together with the $\bJ$ real matrix
replacing the imaginary unit $i$ throughout the theory.
This construction is the natural fermionic counterpart of the real formulation of quantum mechanics and real scalar quantum field theory.

The ordinary Dirac field that a complex four-component spinor
$
\psi(x)\in\mathbb C^4,
$
satisfying the Dirac equation
\begin{equation}
(i\gamma^\mu\partial_\mu-m)\psi=0,
\end{equation}
is 
write as
\begin{equation}
\psi=u+iv, \quad u,v\in\mathbb R^4.
\end{equation}
One defines the real eight-component spinor
\begin{equation}
\bPsi=
\begin{pmatrix}
u\\
v
\end{pmatrix}
\in\mathbb R^8.
\end{equation}
$\bJ$ is 
\be\label{bJ8}
\bJ=
\begin{pmatrix}
0&-I_4\\
I_4&0
\end{pmatrix},
\ee

Realification of $4\times4$ $\gamma$ matricies is the following. One represents them as 
$
\gamma^\mu
=
A^\mu+iB^\mu,
$
where $A^\mu$ and $B^\mu$ are real $4\times4$ matrices and 
\begin{equation}
\bGamma^\mu\equiv \cR(\gamma^\mu)
=
\begin{pmatrix}
A^\mu&-B^\mu\\
B^\mu&A^\mu
\end{pmatrix}.
\end{equation}
The matrices $\bGamma^\mu$ satisfy the Clifford algebra
\begin{equation}
\bGamma^\mu\bGamma^\nu
+
\bGamma^\nu\bGamma^\mu
=
2\eta^{\mu\nu}I_8.
\end{equation}
and commute with $\bJ$ matrix given by \eqref{bJ8}. The real Dirac equation takes a form
\be\label{Dir4}
(\bJ\bGamma^\mu\partial_\mu-m)\bPsi=0,
\ee
that is an entirely real first-order system.

The realification of the Lorentz Transformations is the following. 
If
$
S(\Lambda)\in GL(4,\bC)
$
is the usual spin representation, then its realification is
\begin{equation}
\bS(\Lambda)
=
\begin{pmatrix}
\mathrm{Re}\,S&-\mathrm{Im}\,S\\
\mathrm{Im}\,S&\mathrm{Re}\,S
\end{pmatrix}.
\end{equation}
It is oubvious that 
$
[\bS(\Lambda),\bJ]=0
$.
The real Dirac adjoint is
\be
\overline{\bPsi}
=
\bPsi^T\bGamma^0.
\end{equation}
The standard bilinears become
$
\bar\psi\psi
=
\overline{\Psi}\Psi$,
$
\bar\psi\gamma^\mu\psi
=
\overline{\Psi}\Gamma^\mu\Psi,
$ etc.

The covariant derivative
\begin{equation}
D_\mu
=
\partial_\mu
+
ieA_\mu \longmapsto
\partial_\mu
+
e\bJ A_\mu.
\end{equation}

The Dirac equation takes the form
\begin{equation}
(\bJ\,\bGamma^\mu D_\mu-m)\bPsi=0.
\end{equation}

Gauge transformations become
\be
\bPsi
\longrightarrow
e^{bJ\alpha(x)}
\bPsi,\qquad\mbox{
where}\quad 
e^{\mathbf \bJ\alpha}
=
\cos\alpha
+
\bJ\sin\alpha.
\ee

The standard equal-time anticommutation relations
\begin{equation}
\{
\psi_\alpha(\mathbf x),
\psi_\beta^\dagger(\mathbf y)
\}
=
\delta_{\alpha\beta}
\delta(\mathbf x-\mathbf y)
\end{equation}
become
\begin{equation}
\{
\bPsi_A(\mathbf x),
\bPsi_B(\mathbf y)
\}
=
\delta_{AB}
\delta(\mathbf x-\mathbf y),\quad A,B=1,\ldots,8.
\end{equation}

The realified chirality operator is
\begin{equation}
\bGamma^5
=
\mathcal R(\gamma^5)
=
\bJ\bGamma^0\bGamma^1\bGamma^2\bGamma^3.
\end{equation}
The real chiral projectors are
\begin{equation}
P_L^{\mathbb R}
=
\frac12\left(I-\bGamma^5\right),
\end{equation}
and
\begin{equation}
P_R^{\mathbb R}
=
\frac12\left(I+\bGamma^5\right).
\end{equation}
Since
\begin{equation}
[\bGamma^5,\mathbf J]=0,
\end{equation}
the left- and right-handed realified subspaces are preserved by the complex
structure.
Therefore,
\begin{equation}
S_{D,\mathbb R}
=
S_{L,\mathbb R}
\oplus
S_{R,\mathbb R}.
\end{equation}
Each chiral subspace has twice its usual complex dimension when regarded over
$\mathbb R$, but it contains exactly the same physical degrees of freedom.

\subsubsection*{Remark on Chirality}

The Standard Model is a chiral gauge theory. Left-handed and right-handed
fermions transform differently under $SU(2)$.
In a four-component notation, the usual chiral projectors are
\begin{equation}
P_L
=
\frac12\left(I-\gamma^5\right),\quad P_R
=
\frac12\left(I+\gamma^5\right)
\end{equation}
Their realifications are
\begin{equation}
\mathbb{P}_L
=
\mathcal{R}(P_L),
\quad
\mathbb{P}_R
=
\mathcal{R}(P_R).
\end{equation}
They satisfy
\begin{equation}
\mathbb{P}_L^2=\mathbb{P}_L,
\quad
\mathbb{P}_R^2=\mathbb{P}_R,
\quad
\mathbb{P}_L\mathbb{P}_R=0,\quad
\mathbb{P}_L+\mathbb{P}_R=\bI.
[\mathbb{P}_L,\bJ]=0,\quad [\mathbb{P}_R,\bJ]=0.
\end{equation}
The chirality conditions become
\begin{equation}
\mathbb{P}_L\mathbb{Q}_L
=
\mathbb{Q}_L,
\quad
\mathbb{P}_L\mathbb{L}_L
=
\mathbb{L}_L,
\end{equation}
and
\begin{equation}
\mathbb{P}_R\mathbb{U}_R
=
\mathbb{U}_R,
\quad
\mathbb{P}_R\mathbb{D}_R
=
\mathbb{D}_R,
\quad
\mathbb{P}_R\mathbb{E}_R
=
\mathbb{E}_R.
\end{equation}
Thus realification preserves the chiral structure of the Standard Model.

\newpage
\subsection{Realification of Majorana Fermions in Four-Dimensional Space-Time}
The Majorana fermion is especially natural in a real formulation of quantum
field theory. Unlike a general Dirac fermion, a Majorana fermion can already
be represented by a real four-component spinor in a suitable representation
of the Clifford algebra. Nevertheless, one may also formulate the theory in a
representation-independent way by realifying the complex Dirac spinor and
imposing the Majorana condition, see below \eqref{Mc}, as a real linear constraint.

In the standard formulation,
a Majorana spinor is a Dirac spinor satisfying the Majorana condition
\be\label{Mc}
\psi^C=\psi,\quad \mbox{where}
\quad
\psi^C=C\bar{\psi}^{\,T},
\ee
and the charge-conjugation matrix $C$ satisfies
\be
C\gamma^\mu C^{-1}
=
-\left(\gamma^\mu\right)^T.
\ee

In four-dimensional Minkowski space  there exists a Majorana representation 
in which the gamma matrices are purely imaginary:
\begin{equation}
\left(\gamma^\mu\right)^*
=
-\gamma^\mu.
\end{equation}

Therefore the matrices $i\gamma^\mu$ are real. Writing
\begin{equation}
\gamma^\mu=-iG^\mu,
\end{equation}
where $G^\mu$ are real matrices, the Dirac equation becomes
\begin{equation}
\left(G^\mu\partial_\mu-m\right)\psi=0.
\end{equation}
All coefficients in this equation are real. Consequently, the Majorana field
may be chosen as $
\psi(x)\in\mathbb{R}^4.$
Thus, in a Majorana basis, no doubling of the spinor space is necessary.

In a Majorana basis, the four real field components may be normalized so that
the equal-time canonical anticommutation relations take the form
\begin{equation}
\left\{
\psi_\alpha(t,\mathbf{x}),
\psi_\beta(t,\mathbf{y})
\right\}
=
\delta_{\alpha\beta}
\delta^{(3)}(\mathbf{x}-\mathbf{y})
\end{equation}

\newpage
\subsection{Realification of Weyl Fermions in Four-Dimensional Space-Time}\label{app:RealWeyl}
\subsubsection*{Weyl Spinors in 4 dimensional case}

The realification of Weyl fermions is subtler than the corresponding
construction for Dirac and Majorana fermions. A Weyl spinor is intrinsically
chiral and transforms in one of the two inequivalent complex spinor
representations of the Lorentz group. A single Weyl representation does not
admit a Lorentz-invariant reality condition. Nevertheless, it can be
reformulated as a real four-dimensional vector space equipped with a
distinguished complex structure.
\\

The connected Lorentz group $SO^+(1,3)$ has universal covering group
$
\operatorname{Spin}^+(1,3)\cong SL(2,\mathbb C).
$
Its two irreducible complex spinor representations are
$
\left(\frac12,0\right)$
and
$\left(0,\frac12\right)$.
A left-handed Weyl spinor is denoted by
$
\chi_\alpha(x)$,
$\alpha=1,2$
and transforms in
$
\left(\frac12,0\right)$.
A right-handed Weyl spinor is denoted by
$ \bar\eta^{\dot\alpha}(x),\,
\dot\alpha=1,2$,
and transforms in
$\left(0,\frac12\right)$.
Each Weyl spinor has two complex components, or equivalently four real
components.
\\

For a massless left-handed Weyl fermion, the equation of motion is
\be\label{Meq}
i\bar\sigma^\mu\partial_\mu\chi=0,\quad \mbox{where}\quad
\bar\sigma^\mu=(I_2,-\sigma^1,-\sigma^2,-\sigma^3).
\end{equation}
For a right-handed Weyl fermion, one may write
\begin{equation}
i\sigma^\mu\partial_\mu\bar\eta=0,\quad \mbox{where}\,
\sigma^\mu=(I_2,\sigma^1,\sigma^2,\sigma^3).
\end{equation}

The Pauli matrices are
\begin{equation}
\sigma^1=
\begin{pmatrix}
0&1\\
1&0
\end{pmatrix},
\qquad
\sigma^2=
\begin{pmatrix}
0&-i\\
i&0
\end{pmatrix},
\qquad
\sigma^3=
\begin{pmatrix}
1&0\\
0&-1
\end{pmatrix}.
\end{equation}

\subsubsection*{Realification of the Weyl Spinors}
Writing the complex left-handed Weyl spinor as
\begin{equation}\label{W1}
\chi=u+iv,\quad u,v\in\mathbb R^2,
\end{equation}
we
define the realified spinor
 $\Xi$
\begin{equation}\label{W2}
\Xi=\cR(\chi)=
\begin{pmatrix}
u\\
v
\end{pmatrix}
\in\mathbb R^4.
\end{equation}

Since for a complex $2\times2$ matrix
\be
M=A+iB,
\end{equation}
where $A$ and $B$ are real matrices,  its realification  is defined by
\begin{equation}
\mathcal R(M)
=
\begin{pmatrix}
A&-B\\
B&A
\end{pmatrix},
\end{equation}
the realifications of the identity and Pauli matrices are
\bea\nn
\mathcal R(I_2)&=&I_4,\quad
\mathcal R(\sigma^1)
=
\begin{pmatrix}
0&1&0&0\\
1&0&0&0\\
0&0&0&1\\
0&0&1&0
\end{pmatrix},
\quad
\mathcal R(\sigma^2)
=
\begin{pmatrix}
0&0&0&1\\
0&0&-1&0\\
0&-1&0&0\\
1&0&0&0
\end{pmatrix},
\quad
\mathcal R(\sigma^3)
=
\begin{pmatrix}
1&0&0&0\\
0&-1&0&0\\
0&0&1&0\\
0&0&0&-1
\end{pmatrix}.
\\\label{Sigma}
\eea
All these matrices are real and commute with $\bJ$.
\\

The left-handed Weyl equation \eqref{Meq}
becomes
\begin{equation}
\bJ\,
\overline{\Sigma}^{\mu}
\partial_\mu\Xi_L=0.
\end{equation}
Since $\bJ$ is invertible, this is equivalent to
\begin{equation}
\overline{\Sigma}^{\mu}
\partial_\mu\Xi_L=0.
\end{equation}
For a right-handed Weyl spinor, the corresponding real equation is
\begin{equation}
\bJ\,
\Sigma^\mu
\partial_\mu\Xi_R=0.
\end{equation}
Writing
\begin{equation}
\bar\sigma^\mu=A^\mu+iB^\mu.
\end{equation}
we define
\begin{equation}
\overline{\Sigma}^{\mu}
=
\begin{pmatrix}
A^\mu&-B^\mu\\
B^\mu&A^\mu
\end{pmatrix}.
\end{equation}
The equation
\begin{equation}
\overline{\Sigma}^{\mu}\partial_\mu\Xi=0
\end{equation}
is equivalent to
\begin{align}
A^\mu\partial_\mu u
-
B^\mu\partial_\mu v
&=0,
\\
B^\mu\partial_\mu u
+
A^\mu\partial_\mu v
&=0.
\end{align}

These equations are exactly the real and imaginary parts of
\begin{equation}
\bar\sigma^\mu\partial_\mu\chi=0.
\end{equation}
Hence the realified and complex formulations are mathematically equivalent.
\\

A left-handed Weyl spinor transforms as
\begin{equation}
\chi'(x')=S_L(\Lambda)\chi(x),
\quad \mbox{
where}\quad 
S_L(\Lambda)\in SL(2,\mathbb C).
\end{equation}
Its realification is
\be
\bS_L(\Lambda)
=
\mathcal R\bigl(S_L(\Lambda)\bigr)
=
\begin{pmatrix}
\operatorname{Re}S_L&-\operatorname{Im}S_L\\
\operatorname{Im}S_L&\operatorname{Re}S_L
\end{pmatrix}.
\ee

The realified spinor transforms according to
\begin{equation}
\Xi'(x')
=
\bS_L(\Lambda)\Xi(x).
\end{equation}
Moreover,
$[\bS_L(\Lambda),\bJ]=0$.
Note, that 
 a single Weyl spinor admits no Lorentz-invariant Majorana reality
condition in four-dimensional Minkowski space.
\\

The standard kinetic Lagrangian for a left-handed Weyl fermion is
\begin{equation}\label{LW}
\mathcal L_{\mathrm W}
=
i\chi^\dagger\bar\sigma^\mu\partial_\mu\chi.
\end{equation}

Because the field is Grassmann-valued, the realified action should be
constructed by realifying the Hermitian bilinear form rather than by simply
replacing every complex expression by an ordinary commuting real quadratic
form.
Let
\begin{equation}
\Xi=
\begin{pmatrix}
u\\
v
\end{pmatrix}.
\end{equation}
The Hermitian product on $\mathbb C^2$ induces a real Euclidean form
\begin{equation}
g(\Xi_1,\Xi_2)
=
\operatorname{Re}
\left(
\chi_1^\dagger\chi_2
\right),
\end{equation}
and a symplectic form
\begin{equation}
\omega(\Xi_1,\Xi_2)
=
g(\bJ\Xi_1,\Xi_2).
\end{equation}

The Weyl kinetic term may then be represented in real notation by the
corresponding real bilinear obtained from the RHS of \eqref{LW}.
Symbolically, one may write
\begin{equation}
\mathcal L_{\mathrm W}^{\cR}
=
g\left(
\Xi,
\bJ\,\overline{\Sigma}^{\mu}\partial_\mu\Xi
\right),
\end{equation}
with the understanding that $\Xi$ is Grassmann odd. 


\subsubsection*{Realification of the Gauge Interactions}
For an Abelian gauge interaction, the covariant derivative is
\begin{equation}
D_\mu
=
\partial_\mu+iqA_\mu.
\end{equation}
After realification,
\begin{equation}
\bD_\mu^{U(1)}
=
\partial_\mu+qA_\mu\bJ.
\end{equation}

The gauge transformation
$
\chi\rightarrow e^{iq\alpha(x)}\chi$
becomes
$
\Xi
\rightarrow
e^{q\alpha(x)\bJ}\Xi.
$
The group $U(1)$ acts as real orthogonal rotations preserving $\bJ$.

For a non-Abelian group with complex generators $T^a$, one replaces each
generator by its realification
$\bT^a=\mathcal R(T^a)$
and the realified covariant derivative  then has the form
\begin{equation}
\bD_\mu
=
\partial_\mu
+
gA_\mu^a\,
\mathcal R(iT^a).
\end{equation}

After realification, the usual Standard Model covariant derivative acting on Weyl fermions,
\begin{equation}
D_\mu
=
\partial_\mu
-i g_s G_\mu^a T_c^a
-i g W_\mu^i T_L^i
-i g'Y B_\mu,
\end{equation}
takes the form
\begin{equation}\label{covDer}
\bD_\mu
=
\partial_\mu
+
g_sG_\mu^a\mathbb T_c^a
+
gW_\mu^i\mathbb T_L^i
+
g'B_\mu\mathbb Y,
\end{equation}
where
$\mathbb T^a=\mathcal R(-iT^a)$ are real antisymmetric matrices and $\bY=Y\bJ$
denotes the realified 
$U(1)$ hypercharge generator.

Applying this to the left-handed Weyl fermionic kinetic term,
\begin{equation}
\cL_{kin}
=
i\psi^\dagger\bar\sigma^\mu D_\mu\psi,
\end{equation}
yields
 \be
\cL_{kin}^{\cR}
=
g\left(
\Xi,
\bJ\,\overline{\Sigma}^{\mu}\partial_\mu\Xi
\right),\ee
where $
\overline{\Sigma}^{\mu}
=
\mathcal R(\bar\sigma^\mu).$ The kinetic operator is therefore a real matrix differential operator.
\\

\end{document}